\documentclass[12pt,colorinlistoftodos]{article} % Don't use a specific style unless really required

\usepackage{amsmath,amsfonts,amssymb,amsthm}

\usepackage[hidelinks]{hyperref}
\usepackage{booktabs}
\usepackage{enumitem}
\usepackage{graphicx}
\usepackage{subcaption}

\usepackage{todonotes}

\newcommand{\bm}{\boldsymbol}
\usepackage[
  letterpaper,
  left=2.5cm,
  right=2.5cm,
  top=2.5cm,
  bottom=2.5cm,
  headheight=0.8cm,
  headsep=0.8cm
]{geometry}
\PassOptionsToPackage{force}{filehook}
\usepackage{tikz}
\usetikzlibrary{shapes,arrows}
\usetikzlibrary{positioning}
\tikzstyle{cloud} = [draw, ellipse,fill=red!20, node distance=0.87cm,
minimum height=2em]
\tikzstyle{line} = [draw, -latex']
\usetikzlibrary{shapes.symbols,shapes.callouts,patterns}
\usetikzlibrary{calc}
\usepackage{pgfplots}
\usetikzlibrary{decorations.pathreplacing}
\pgfplotsset{compat=1.15}
\def\diag{\mathsf{diag}}

\def\bX{\mathbf{X}}
\def\b0{\mathbf{0}}

\def\IR{\mathbb{R}}

\def\R{\mathcal{R}}

\usepackage{bm}

\def\bv{\bm{v}}
\def\bx{\bm{x}} 
\def\by{\bm{y}}
\def\bw{\bm{w}}

\def\bE{\bm{E}}

\def\bX{\bm{X}}

\def\b0{\bm{0}}

\newcommand{\keywords}[1]{\par\noindent\textbf{Keywords :} #1}

\selectfont
\theoremstyle{plain}
\newtheorem{theorem}{Theorem}[section]
\newtheorem{lemma}[theorem]{Lemma}
\newtheorem{proposition}[theorem]{Proposition}

\newtheorem{remark}[theorem]{Remark}

\numberwithin{equation}{section}
\usepackage{amsmath}
\usepackage{todonotes}
\usepackage{booktabs}
\usepackage{enumitem}
\usepackage{geometry}
\begin{document}
\title{A model for cholera with infectiousness of deceased individuals and vaccination}
\author{Annour Saad Abdramane$^{\bf a}$ \and Julien Arino$^{\bf b}$ \and Patrick Mimphis Tchepmo Djomegni$^{\bf c}$\and Mahamat Saleh Daoussa Haggar$^{\bf a}$ \\
\small $^{\bf a}$Laboratory L2MIAS, University of N'Djamena, Chad\\
\small	$^{\bf b}$Department of Mathematics, University of Manitoba, Winnipeg, MB, Canada\\
\small	$^{\bf c}$School of Mathematical and Statistical Sciences, North-West University, South Africa}

\date{\today}

\maketitle
\begin{abstract}
    A cholera transmission model is formulated that incorporates waterborne and horizontal transmissions as well as infectivity of deceased individuals. 
    The model includes an Allee effect for the bacteria in the environment and imperfect and waning vaccination.
    Mathematical properties of the model are investigated, with an environmental bistability shown to combine with a vaccine-driven one, although a computational search for the latter fails to detect its presence in realistic parameter ranges.
    The computational analysis also considers the interplay between vaccination strategy, vaccine efficacy and waning, as well as the effect of transmission of the disease during funeral rites.
    The effects of control scenarios such as water, sanitation and hygiene (WASH) or safe and dignified burials are assessed.
\end{abstract}

%%%%%%%%%%%%%%%%%%%%%%
%%%%%%%%%%%%%%%%%%%%%%
%%%%%%%%%%%%%%%%%%%%%%
%%%%%%%%%%%%%%%%%%%%%%
\section{Introduction}
\label{sec:introduction}

Cholera is an acute diarrhoeal infection that is a global public health threat causing devastating epidemics, particularly in regions where access to clean water and sanitation infrastructure is limited \cite{wang2022mathematical}.
For instance, cholera triggered large-scale outbreaks starting in 2010 in Haiti and in late 2016 in Yemen, with the Yemen epidemic alone accounting for more than 2.5 million cumulative suspected cases between late 2016 and November 2021 \cite{federspiel2018cholera,wang2022mathematical}.
It also remains endemic in certain regions of sub-Saharan Africa and Asia.
It is estimated that approximately 1.3 billion people live in areas at risk of cholera \cite{mondiale2017cholera}.
Each year, these regions are estimated to experience around 2.9 million cases, with nearly 95,000 deaths \cite{ali2015updated,ganesan2020cholera}.

Cholera is caused by the bacterium \textit{Vibrio cholerae}, primarily from serogroups O1 and less frequently O139 \cite{piret2021pandemics, world2010cholera}.
It primarily inhabits aquatic environments such as estuaries, rivers and groundwater, particularly in areas contaminated by human waste \cite{finkelstein1996cholera}.
Humans are the only known natural vertebrate hosts in which \textit{V. cholerae} causes clinical disease, although the bacterium also colonises aquatic organisms such as copepods and chironomids \cite{colwell1996global,lutz2013environmental}.
The complex dynamics of cholera in the wild are deeply tied to human excretion and the environment.
There is evidence \cite{colwell1996global,xu1982survival} that the environment can be a reservoir of \textit{V. cholerae} for long periods of time (spanning months to several years between outbreaks), often in a viable but non-culturable state, alongside known cases where culturable cholera was not detectable after an outbreak had terminated \cite{islam1990survival}.
Altogether, the situation is highly location- and context-dependent.

Mathematical modelling of cholera remains less developed compared to other infectious diseases \cite{wang2022mathematical}.
The earliest model of waterborne disease we are aware of is \cite{capasso1979mathematical}, which assumes a loop between humans and an environmental reservoir where \textit{V. cholerae} can survive and proliferate.
In \cite{codecco2001endemic}, a more detailed description of the infection in humans was introduced, with the model being specific to cholera.
This was generalised in \cite{pascual2002cholera} by incorporating a fourth equation representing the volume of water in which pathogens develop.
In \cite{joh2009dynamics}, the dynamics of infectious diseases whose primary mode of transmission is indirect and mediated by contact with a contaminated reservoir was studied.
Most mathematical models of cholera follow \cite{codecco2001endemic} and assume that cholera is introduced into the environment by contaminated humans and eliminated from it by pathogen death.
Few models deviate from the formulation of \cite{codecco2001endemic}.
Some assume logistic dynamics \cite{wang2012cholera} or distinguish between hyperinfective and lower-infective bacterial states \cite{neilan2010mathematical}.
Even less frequent are models encoding an Allee effect \cite{Capone_2020, Kolaye_Guilsou_2023, Yang2020}, a positive relationship between population density and the \emph{per capita} growth rate at low densities \cite{courchamp2008allee,stephens1999what}, which is the route we follow here.
Specifically, a \emph{strong} Allee effect introduces a critical threshold density below which the population decays and above which it can grow, allowing us to cover both the case where cholera becomes extinct (locally) in the absence of new human contamination of the environment and the one where cholera is able to sustain itself intrinsically.

Beyond environmental exposure, cholera has multiple transmission routes that present unique challenges.
It is mostly transmitted through the faecal-oral route, i.e., the ingestion of water or food contaminated with the faeces of an infected person.
Direct person-to-person transmission is rare but happens \cite{muzembo2022cholera}.
Clinical studies indicate that individuals who die from cholera remain infectious for some time.
Transmission has been shown to occur during the washing or transportation of bodies for funeral rites as well as during events peripheral to funerals when the deceased bodies are still present \cite{gunnlaugsson1998funerals}.
Although the infection risk posed by dead bodies after natural disasters is often overstated in general \cite{conly2005natural, morgan2004infectious}, corpses of cholera victims pose a documented transmission hazard during epidemics and funeral handling \cite{cardoso2016study, gunnlaugsson1998corpses, gunnlaugsson1998funerals, ifrc2018safe}.
Modelling of this corpse-based transmission route is not frequent.
The few examples that we know of concern transmission of the Ebola virus \cite{adefisan2018mathematical, tasse2024metapopulation}.
While the possibility of corpse-mediated cholera transmission was noted (though neglected) in \cite{che2021mathematical}, to the best of our knowledge this mechanism has not been explicitly analysed in a mathematical model of cholera.

The prevention of cholera relies mainly on hygiene measures, including access to safe drinking water, adequate sanitation and strict personal hygiene practices, the so-called water, sanitation and hygiene (WASH) interventions \cite{taylor2015impact}.
However, vaccination also plays a critical role, even though the cholera vaccine is not part of the standard immunisation schedule.
It is recommended for individuals at higher risk of infection, including travellers to cholera-endemic areas, populations living in poor sanitary conditions and refugees in camps where a cholera outbreak could occur.
Oral cholera vaccines (both inactivated and live attenuated) provide significant protection, although vaccine effectiveness is often lower in young children and in low-resource endemic settings \cite{leung2022optimizing, marie2018pathophysiology}.
The oral cholera vaccine (OCV) is typically administered as a two-dose regimen and is less effective in children under five years old than in those aged five years and older \cite{leung2022optimizing}.
Because their effectiveness is time-limited, vaccines are generally used as a complementary measure to hygiene interventions \cite{finkelstein1996cholera}.
In contrast to cholera, modelling of vaccination has a rich history in mathematical epidemiology, dating back all the way to the epidemiological model of Daniel Bernoulli \cite{Bernoulli1760}, which included inoculation.
In its simplest form, vaccination is modelled by assuming that it renders individuals completely immune, moving them directly from the susceptible into the recovered compartment (as done for cholera, e.g., in \cite{tuite2011cholera}).
However, because of the homogeneity of content assumption of compartmental models \cite{JacquezSimon1993}, this way of proceeding makes vaccinated individuals indistinguishable from recovered ones.
As a consequence, later vaccination models introduced a specific compartment for vaccinated individuals.
Starting with the SIS model with vaccination of \cite{Kribs-ZaletaVelasco-Hernandez2000}, there came the realisation that imperfect and waning vaccination could imply more complicated dynamics, namely subcritical endemic equilibria.
This phenomenon, which is typically called a backward bifurcation, had already been characterised in various epidemiological models (including core-group, disease-induced mortality and treatment models) \cite{dushoff1998backwards, hadeler1997backward, huang1992stability}.
It was formalised in an SIRS model with vaccination by \cite{arino2003global} and studied in multiple other papers \cite{ArinoCookeVdDVelasco2004,greenhalgh2001mathematical, gumel2012causes}.
When a backward bifurcation is present, the disease can persist endemically even when the reproduction number in the presence of vaccination is less than 1.
Regarding cholera itself, models with vaccination have been considered for a while.
Besides the already cited \cite{tuite2011cholera}, \cite{ghosh2015modeling} studied the effect of seasonality in pathogen transmission on vaccination strategies under several disease scenarios.
Their model is an extension of the SIWR model in \cite{tien2010multiple}.
Other cholera models with vaccination include \cite{alarydah2013modeling, onuorah2022mathematical, sun2017transmission, zhou2011modeling}.
Several cholera models with vaccination exhibit a backward bifurcation driven by saturated treatment or recovery \cite{abdul2021stability, sharma2021, sharma2021backward, sharma2021bifurcation}.

The objective of our study is therefore to extend a classic cholera model by incorporating an Allee effect, modelling vaccination dynamics and considering deceased individuals as a potential source of contamination.
This comprehensive approach allows us to evaluate the extent to which such a contamination route is hazardous in the presence of targeted control measures.

The remainder of this work is structured as follows. In Section~\ref{sec:mathematical-model}, we derive the mathematical model, while the analytical study of the model is presented in Section~\ref{sec:math-analysis_v}, with some of the results proved in an appendix for legibility. 
Section~\ref{sec:computational-analysis} presents a computational analysis of the model; a Discussion concludes the work.

%%%%%%%%%%%%%%%%%%%%%%%%%%% 
%%%%%%%%%%%%%%%%%%%%%%%%%%% 
%%%%%%%%%%%%%%%%%%%%%%%%%%% 
%%%%%%%%%%%%%%%%%%%%%%%%%%% 
\section{Model formulation} 
\label{sec:mathematical-model}
We start from a classic $SIRW$ model for waterborne transmission of cholera including human-to-human transmission and introduce an additional compartment, $D$, to account for individuals having died from cholera but whose bodies are still infectious.
Furthermore, we consider an explicit compartment for vaccinated individuals, assuming that the vaccine is both imperfect and waning, as was studied for a classic SIRS model in \cite{arino2003global}.
In the sequel, we denote \( N_{H} = S + I + R + V\) the total (live) human population.
See Table~\ref{tab:state-variables} for state variables and Figure~\ref{fig:flow_diagram_vaccination} for the flow diagram of the model.

\begin{table}[htbp]
\centering
\begin{tabular}{cl}
\toprule
Variable & Description \\
\midrule
$S$ & Number of susceptible individuals \\
$I$ & Number of infected and infectious individuals \\
$R$ & Number of recovered immune individuals \\
$V$ & Number of vaccinated individuals \\
$N_H$ & Total number of (live) humans \\
$D$& Number of infectious dead individuals\\
$W$ & Dimensionless concentration of \emph{V. cholerae} in water \\
\bottomrule
\end{tabular}
\caption{State variables of the SIRVDW model.}
\label{tab:state-variables}
\end{table}

Cholera vaccination is typically not administered at birth, being generally used in the context of heightened risks.
As a consequence, we assume that there is no vaccination at birth and that all individuals are born susceptible at the rate $b$.
Individuals in all compartments are subject to natural death at the \emph{per capita} rate $d_H$.
Susceptible individuals are subject to a force of infection $\lambda_S$, while those vaccinated are subject to a force of infection $\lambda_V$; see below for details.
After becoming contaminated, individuals either recover at the \emph{per capita} rate $\gamma$ or die at the \emph{per capita} rate $\delta$.
While they are infectious, individuals shed \emph{V. cholerae} into the environment at the rate $\zeta$.
Furthermore, deceased bodies that have not yet been sanitised or buried continue to shed the pathogen into the environment at a rate $\eta\zeta$, where $\eta$ represents their relative shedding capacity.
A conversion applies to the shedding rate $\zeta$; see details in Appendix~\ref{app:nondimensionalisation-W}.
The recovery rate $\gamma$ incorporates both natural recovery and that induced by treatment, which is not explicitly modelled here. 
As a consequence, the disease-induced death rate $\delta$ incorporates both the high mortality of untreated individuals and the much lower one of treated ones.
Upon recovery into the $R$ compartment, individuals are temporarily immune to reinfection by cholera; this natural immunity wanes at the \emph{per capita} rate $\varepsilon$ (after an exponentially distributed duration with mean $1/\varepsilon$ time units), returning individuals to the susceptible compartment $S$. 
Finally, dead individuals are sanitised or properly buried at the \emph{per capita} rate $e$.

\begin{figure}[htbp]
    \centering
    \begin{tikzpicture}[scale=1.1, 
        every node/.style={transform shape},
        auto,
        box/.style={minimum width={width("N-1")+9pt},
        draw, rectangle}]
        % Variables d'etat
        \node [box, fill=green!30] at (0,0) (S) {$S$};
        \node [box, fill=red!40] at (3,0) (I) {$I$};
        \node [box, fill=yellow!20] at (1.5,-2) (V) {$V$};
        \node [box, fill=yellow!20] at (6,1) (R) {$R$};
        \node [box, fill=red!40] at (6,-1) (D) {$D$};
        \node [box, fill=blue!20] at (3,2) (W) {$W$};
        % Pour les fleches d'entree et sorties
        \coordinate[left=1cm of S] (birthS);
        \coordinate[below of=S] (deathS);
        \coordinate[below=1cm of I] (deathI);
        \coordinate[right=1cm of R] (deathR);
        \coordinate[below=1cm of V] (deathV);
        \coordinate[left=1.3cm of W] (birthW);
        \coordinate[right=1.3cm of W] (deathW);
        \coordinate[right=1cm of D] (deathD);
        %% Birth and death flows
        \path [line, thick] (birthS) to node [midway, above] (TextNode) {$b$} (S);
        \path [line, thick] (birthW) to node [midway, above] (TextNode) {$b_C(W)$} (W);
        \path [line, thick] (W) to node [midway, above] (TextNode) {$d_C(W)$} (deathW);
        \path [line, thick] (I) to node [midway, right] (TextNode) {$d_HI$} (deathI);
        \path [line,  thick] (R) to node [midway, above] (TextNode) {$d_HR$} (deathR);
        \path [line,  thick] (S) to node [midway, left] (TextNode) {$d_HS$} (deathS);
        \path [line, thick] (D) to node [midway, above] (TextNode) {$eD$} (deathD);
        \path [line, thick] (V) to node [midway, right] (TextNode) {$d_HV$} (deathV);
        %% Inter-compartmental flows
        \path [line,   bend left=10,thick] (S) to node [above, sloped] (TextNode) {$vS$} (V);
        \path [line, bend left=10, thick] (V) to node [below, sloped] (TextNode) {$\theta V$} (S);
        \path [line, thick] (S) to node [midway, above, sloped] (TextNode) {$\lambda_S S$} (I);
        \path [line, thick] (V) to node [above, sloped] (TextNode) {$ \lambda_V V$} (I);
        \path [line, thick, dotted] (I) to node [midway,above,sloped] (TextNode) {$\zeta I$} (W);
        \path [line, thick, dotted] (D) to node [near start,above,sloped] (TextNode) {$\eta\zeta D$} (W);
        \path [line, thick] (R) -- (6,2.75) -- node [midway, above, sloped] (TextNode) {$\varepsilon R$} (0,2.75) to (S);
        \path [line, thick] (I) to node [midway, above, sloped] (TextNode) {$\delta I$} (D);
        \path [line, thick] (I) to node [near end, above, sloped] (TextNode) {$\gamma I$} (R);
        \end{tikzpicture}
    \caption{Flow diagram of the SIRVDW model \eqref{sys:general_form_v}.}
    \label{fig:flow_diagram_vaccination}
\end{figure}

Susceptible individuals are vaccinated at the \emph{per capita} rate $v$.
When vaccinated, the force of infection they are subject to is reduced by a factor $1-\sigma$, where $\sigma\in[0,1]$ is the vaccine efficacy, which we assume to be the same regardless of the source of the immunological challenge.
Therefore, $\lambda_V = (1-\sigma) \lambda_S$ is the force of infection acting on vaccinated individuals.
Note that because $\sigma$ acts directly at the individual contact level to reduce susceptibility per exposure (specifically, leaky vaccine efficacy against susceptibility), we refer to $\sigma$ as vaccine efficacy rather than population-level vaccine effectiveness \cite{halloran2010design,MilwidSteriuArinoHeffernanEtAl2016}.
Vaccine protection is lost after some time because of so-called \emph{waning}; the duration of protection by the vaccine is exponentially distributed with mean $1/\theta$ time units.

Susceptible or vaccinated individuals acquire the infection through three pathways: from the water through the faecal-oral route, horizontally by person-to-person transmission or by contact with unsanitised dead bodies, with the transmission coefficients for these different pathways $\beta_W$, $\beta_I$ and $\beta_D$, respectively.
Standard (proportional) incidence is used for direct contamination between humans (alive or dead), while a saturating incidence function $\beta_WW/(1+W)$ is used for contamination by the environment; here $W$ has been nondimensionalised by the half-saturation concentration $\kappa$ so that the half-saturation constant equals $1$ (see Appendix~\ref{app:nondimensionalisation-W}).
As a consequence, the force of infection $\lambda_S$ acting on susceptible individuals is given by \eqref{sys:general_form_FoI} below, where the proportional formulation scales contact rates with living and deceased infectious individuals relative to the living population $N_H$.

\begin{table}[htbp]
\centering
\begin{tabular}{cl}
\toprule
Parameter & Description\\ 
\midrule
\multicolumn{2}{l}{Human-related parameters} \\
$b$ & Recruitment rate of human population \\
$d_H$ & Natural death rate of humans \\
$\beta_{W}$ & Infection by water \\
$\beta_{I}$ & Infection by humans \\
$\beta_{D}$ & Infection from unburied corpses \\
$\gamma$ & Recovery rate \\
$\varepsilon$ & Rate of loss of natural immunity \\
$\delta$ & Disease-induced death rate \\
$e$ & Interment rate \\
\midrule
\multicolumn{2}{l}{Vaccination-related parameters} \\
$v$ & Rate of vaccination \\
$\theta$ & Rate of vaccine protection waning \\
$\sigma$ & Vaccine efficacy\\
\midrule
\multicolumn{2}{l}{\emph{V. cholerae}-related parameters} \\
$\zeta$ & Pathogen shedding rate by humans \\
$\eta$ & Relative pathogen shedding rate by deceased individuals \\
$r_C$ & Demographic rate constant (low-density decay rate) of pathogens \\
$A_C$ & Allee threshold for pathogens \\
$K_C$ & Carrying capacity of pathogens in water \\
\bottomrule
\end{tabular}
\caption{Parameters of the SIRVDW model.}
\label{tab:model-parameters}
\end{table}

Finally, regarding \emph{V. cholerae} in the water, the bacterium is subject to intrinsic population dynamics.
\emph{V. cholerae} is a known inhabitant of brackish riverine, estuarine and coastal waters \cite{almagro2013cholera} and its population dynamics in natural environments are very complex \cite{lutz2013environmental}.
To capture this complexity, we consider an intrinsic growth of the bacteria exhibiting a strong Allee effect.
Such an assumption incorporates both the natural death of bacteria at low concentrations and their capacity to sustain themselves when their concentration crosses the Allee threshold $A_C$. The environment has a carrying capacity $K_C$ and the demographic timescale is governed by $r_C$ (which corresponds to the \emph{per capita} decay rate of the bacteria at low concentrations $W \ll A_C$).

%%%%%%%%%%%%%%%%%%%%%%%%%%%%%%%%%%%%%%%%%
Taking these assumptions into consideration, we obtain the following model,
\begin{subequations}
		\label{sys:general_form_v}
		\begin{align}
			\frac{d}{dt}S &= 
			b+\varepsilon R+\theta V -\lambda_S  S 
			-(v +d_H)S 
			\label{sys:general_form_dSv} \\ 
			\frac{d}{dt}I  &= 
			\lambda_S  S +\lambda_VV  -(\gamma + \delta +d_H)I  
			\label{sys:general_form_dIv} \\ 
			\frac{d}{dt}R  &=
			\gamma  I -(\varepsilon +d_H)R 
			\label{sys:general_form_dRv} \\ 
			\frac{d}{dt}V  &= v S -
			\theta V - \lambda_VV  -d_H  V 
			\label{sys:general_form_dVv}\\
			\frac{d}{dt}D  &= 
			\delta I -e D 
			\label{sys:general_form_dDv} \\
			\frac{d}{dt}W  &= 
			\zeta I + \eta\zeta D + r_C W \left( 1 - \frac{W}{K_C} \right) \left( \frac{W}{A_C} - 1 \right),
			\label{sys:general_form_dWv} 
		\end{align}
where the forces of infection acting on susceptible and vaccinated individuals are
\begin{equation}\label{sys:general_form_FoI}
\lambda_S = \beta_{W}\frac{W}{1+W}+\beta_{I}\frac{I}{N_H} +\beta_{D}\frac{D}{N_H} \quad\text{and}\quad \lambda_V = (1-\sigma)\lambda_S,
\end{equation}
\end{subequations}		
respectively. 
The model is considered  with the initial condition
\begin{equation}\label{sys:general_form_IC}
    S (0) > 0, I (0)\geq 0, R (0) \geq 0, V (0)\geq 0, D (0) \geq 0, W (0) \geq 0.
\end{equation}

While only the net cubic rate $b_C(W) - d_C(W)$ enters the dynamics of \eqref{sys:general_form_dWv}, following classic polynomial formulations of the strong Allee effect \cite{courchamp2008allee,Kolaye_Guilsou_2023} it is algebraically convenient to separate the positive cooperative facilitation term $b_C(W)$ and the loss term $d_C(W)$ shown in Figure~\ref{fig:flow_diagram_vaccination}. 
Expanding the cubic polynomial in \eqref{sys:general_form_dWv}, we can rewrite \eqref{sys:general_form_dWv} as
\begin{equation*}
    \frac{d}{dt}W = \zeta I + \eta\zeta D + b_C(W) - d_C(W),
\end{equation*}
where the positive cooperative growth term is given by
\begin{equation*}
    b_C(W) = r_C \left( \frac{1}{A_C} + \frac{1}{K_C} \right) W^2,
\end{equation*}
and the loss term is given by
\begin{equation*}
    d_C(W) = r_C W + r_C \frac{W^3}{A_C K_C}.
\end{equation*}
Because of the strong Allee effect, the linear term $-r_C W$ in $b_C(W) - d_C(W)$ has a negative coefficient, representing baseline bacterial decay at low concentrations ($W \to 0$).
The cooperative growth $b_C(W)$ increases quadratically with $W$, while the cubic term $r_C W^3/(A_C K_C)$ in $d_C(W)$ captures high-density intraspecific competition, scaled by the product $A_C K_C$ so that net intrinsic growth $b_C(W) - d_C(W)$ vanishes precisely at the Allee threshold $W = A_C$ and the carrying capacity $W = K_C$.

%%%%%%%%%%%%%%%%%%
%%%%%%%%%%%%%%%%%%
%%%%%%%%%%%%%%%%%%
%%%%%%%%%%%%%%%%%%
\section{Mathematical analysis of the model}
\label{sec:math-analysis_v}
To simplify the remainder of the analysis, we denote $\bX=(S, I, R, V, D, W)$ the vector of state variables.
Most proofs as well as some intermediate results are deferred to appendices.

%\subsection{Preliminary considerations}
First remark that it is clear that solutions to \eqref{sys:general_form_v} exist and are unique.
By Lemma~\ref{lemma:solutions-positive}, the positive orthant is positively invariant under the flow of \eqref{sys:general_form_v}.
This can actually be tightened as follows.
\begin{lemma}
\label{lemma:bounded-region}
The region
\begin{equation} \label{eq:domain2}
    \Omega:=\left\{ \bX\in \mathbb{R}_+^6;\; N_H(t)\leq \frac{b}{d_H}, W(t)\leq W_{\max}, D(t)\leq \frac{\delta b}{e d_H} \right\},
\end{equation}
where $W_{\max}$ is the largest positive root of the polynomial $\zeta b (1 + \eta \delta/e)/d_H + r_C W (1 - W/K_C)(W/A_C - 1) = 0$, is attracting and positively invariant under the flow of \eqref{sys:general_form_v}.
\end{lemma}

We now turn our attention to the disease-free equilibrium, which is obtained by solving \eqref{sys:general_form_v} when $I=D=W=0$. 
It is easy to check that 
\begin{equation}\label{eq:DFE_general}
    \bE_0=\left(
        \dfrac{\theta +d_H}{\theta +v +d_H}\;\dfrac{b}{d_H}, 0, 0, \dfrac{v}{\theta +v +d_H}\;\dfrac{b}{d_H}, 0, 0\right)
\end{equation}
is the unique disease-free equilibrium (DFE) for \eqref{sys:general_form_v}.
In the absence of vaccination, the DFE takes the form $\bE_0=(b /d_H,0,0,0,0,0)$.

From \eqref{eq:DFE_general}, we get that at the disease-free equilibrium, a fraction
\begin{equation}\label{eq:fraction-vaccinated-DFE}
\Psi_0^v
=
\dfrac{v}{v +\theta +d_H}
\end{equation}
of individuals in the population is vaccinated, i.e., the \emph{vaccine coverage}. 
This expression is useful when setting parameters, since this is a value that is typically known.

Characterising the stability of the disease-free equilibrium is then important, because it provides some understanding of the behaviour of the system in its ideal (disease-free) state as well as some useful relationships between parameters. 
To do this, we compute the reproduction numbers of \eqref{sys:general_form_v} using the method of \cite{VdDWatmough2002}.

Consider the infected compartments $(I,D,W)$; let $\mathcal{F}$ be the vector representing new infections or bacterial growth coming into and $\mathcal{U}$ other flows within or out of these compartments (with a negative sign), respectively, i.e.,
\[
\mathcal{F}=\begin{pmatrix}
    \lambda_S  S+ \lambda_V V \\ 0 \\ b_C(W)
\end{pmatrix}
\text{ and } \mathcal{U}=\begin{pmatrix}
    (\gamma +\delta +d_H)I \\ -\delta I +e  D  \\-\zeta  I - \eta\zeta D + d_C(W)
\end{pmatrix}.
\]
Let $F=\left.\partial \mathcal{F}/\partial (I, D, W)\right|_{\bE_0}$ and $U=\left.\partial \mathcal{U}/\partial (I, D, W)\right|_{\bE_0}$.
Because the cooperative bacterial growth $b_C(W)$ is purely quadratic in $W$, $\left.\partial b_C/\partial W\right|_{W=0}=0$, whereas $\left.\partial d_C/\partial W\right|_{W=0}=r_C$.
Note that while cholera models with linear or logistic bacterial growth in water can yield different expressions for the reproduction number depending on whether environmental bacterial multiplication is placed in $\mathcal{F}$ or $-\mathcal{U}$ \cite{lupica2020computation}, under a strong Allee effect the absence of a linear birth term at $W=0$ implies that both decompositions yield the same matrices $F$ and $U$.
Then
\[
F=\begin{pmatrix}
    \beta_{I}\frac{S^{\star}+(1-\sigma) V^{\star}}{N_H^\star} & \beta_{D}\frac{S^{\star}+(1-\sigma) V^{\star}}{N_H^\star} & \beta_{W} (S^{\star}+(1-\sigma) V^{\star})\\ 0 & 0 & 0 \\0 & 0 & 0
\end{pmatrix}, 
\]
and 
\[
U=\begin{pmatrix}
    \gamma +\delta +d_H  & 0 & 0\\ -\delta & e & 0 \\-\zeta & -\eta\zeta &  r_C
\end{pmatrix}.
\]
The reproduction number in the presence of vaccination $\R_v$ (also referred to as the vaccinated reproduction number) is then the spectral radius of the next generation matrix $FU^{-1}$,  i.e., 
\begin{equation}\label{eq:R0_vaccination}
    \R_{v}= \frac{\left(e r_C\beta_{I}+\delta r_C\beta_{D}+\zeta (e+\eta\delta) \beta_{W}\frac{b}{d_H}\right)\left(\theta +(1-\sigma) v+d_H\right)}{e r_C(\gamma +\delta +d_H)\left(\theta +v +d_H\right)}.
\end{equation}
In the absence of vaccination ($v=0$), the reproduction number takes the form
 \begin{equation} \label{eq:R0}
    \R_{0}= \frac{e r_C\beta_{I}+\delta r_C \beta_{D}+\zeta (e+\eta\delta) \beta_{W} \frac{b}{d_H}}{er_C\left(\gamma +\delta +d_H \right)}.
\end{equation}

Notice that at any equilibrium, \eqref{sys:general_form_dDv} gives $D^\star = \delta I^\star/e$, so that the total shedding influx $\zeta I^\star + \eta\zeta D^\star$ in \eqref{sys:general_form_dWv} (as well as the waterborne contribution in \eqref{eq:R0_vaccination} and \eqref{eq:R0}) can be written as $\zeta_{\text{eff}} I^\star$, where the effective shedding rate is
\begin{equation} \label{eq:zeta_eff}
    \zeta_{\text{eff}} = \zeta \left(1 + \eta \frac{\delta}{e}\right).
\end{equation}

It follows from \eqref{eq:R0_vaccination} and \eqref{eq:R0} that
\begin{equation}\label{eq:Rv_fct_R0}
     \R_{v}= \R_{0}\left(\dfrac{\theta +(1-\sigma) v+d_H}{\theta +v +d_H}\right).
\end{equation}
Since $\sigma \in [0, 1]$, it follows that $\R_v\leq \R_0$, with the inequality being strict whenever vaccination is active and effective ($\sigma>0$ and $v>0$).

Not unexpectedly because of the structural similarity of the human component of the model to systems such as the ones in \cite{arino2003global,arino2022bistability}, \eqref{sys:general_form_v} can exhibit a further \emph{backward bifurcation}.
We defer discussion of this issue to Appendix~\ref{app:vaccine-backward-bifurcation}; suffice it to say here that even if it is a possibility, it does not occur under usual parameter regimes.
To study the endemic equilibria of \eqref{sys:general_form_v} (and of the model without vaccination in Appendix~\ref{app:system-without-vaccination}), we solve the equilibrium equations $S' = R' = V' = D' = 0$ to express $(S, R, V, D)$ in terms of $(W, I)$. 
Substituting these expressions into $W' = 0$ and $I' = 0$ reduces the 6-dimensional equilibrium problem to the intersection of two nullcline curves in the $(W, I)$ plane: the bacterial nullcline $I = I_C(W)$ and the human nullcline $I = I_H(W)$ (see Appendices~\ref{app:system-without-vaccination} and \ref{app:equilibria} for full details). 
Furthermore, substituting these equilibrium expressions for $(S, R, V, D)$ into the differential equations \eqref{sys:general_form_dWv} and \eqref{sys:general_form_dIv} for $W'$ and $I'$ yields a 2-dimensional system in $(W, I)$, which we refer to as the \emph{quasi-steady-state planar reduction} of \eqref{sys:general_form_v}.

The following proposition summarises what is known about the equilibria of \eqref{sys:general_form_v} under the assumption that the vaccine does not induce a backward bifurcation.
\begin{proposition}\label{prop:full-vaccination-equilibria-forward}
Assuming that the human-to-human reproduction number satisfies $\R_0^H := (\beta_I + \beta_D\delta/e)/(\gamma + \delta + d_H) < 1$ and that there is no vaccine-induced backward bifurcation (i.e., $a < 0$ in Theorem~\ref{thm:stability-EEP}), let $\R_v^{SN} > 1$ denote the critical saddle-node bifurcation value of $\R_v$ (varied via $\beta_W$) at which the nullclines $I_H(W)$ and $I_C(W)$ are tangent in $(0, A_C)$ (see Appendix~\ref{app:subsec:full-EEP-lt-A_C}). Then \eqref{sys:general_form_v} possesses the following equilibria.
\begin{enumerate}
    \item There always exists a unique disease-free equilibrium $\bE_{0}$, which is locally asymptotically stable if $\R_v < 1$ and unstable if $\R_v > 1$.
    \item There always exists a unique endemic equilibrium $\bE_m$ with $W^\star > K_C$, which is locally asymptotically stable in the quasi-steady-state planar reduction.
    \item If $\R_v < 1$, there exists one subthreshold endemic equilibrium (with respect to the environmental Allee threshold, $W^\star \in (0, A_C)$), $\bE_u$, which is unstable in the full system.
    \item If $1 < \R_v < \R_v^{SN}$, there exist two subthreshold endemic equilibria with $W^\star \in (0, A_C)$: an equilibrium $\bE_\ell$ (which is locally asymptotically stable in the full system near $\R_v = 1$ and in the quasi-steady-state planar reduction) and an equilibrium $\bE_u$, which is unstable in the full system.
    \item If $\R_v = \R_v^{SN}$, there exists one subthreshold endemic equilibrium $\bE_{SN}$ with $W^\star \in (0, A_C)$, which is a non-hyperbolic saddle-node of the full system.
    \item If $\R_v > \R_v^{SN}$, there are no subthreshold endemic equilibria with $W^\star \in (0, A_C)$.
\end{enumerate}
\end{proposition}

\begin{proof}
    Uniqueness and local stability of the disease-free equilibrium are established in Appendix~\ref{app:subsec:full-model-DFE}. 
    Existence and stability of the macroscopic equilibrium $\bE_m$ are proved in Lemma~\ref{lemma:full-one-large-W-EEP}.
    The properties of the subthreshold equilibria are respectively established by Lemmas~\ref{lemma:full-subthreshold-Eu}, \ref{lemma:full-subthreshold-two-equilibria}, \ref{lemma:full-subthreshold-tangent} and \ref{lemma:full-subthreshold-none}.
\end{proof}

The situation described in Proposition~\ref{prop:full-vaccination-equilibria-forward} is shown in Figure~\ref{fig:forward-bifurcation}.
To interpret this situation epidemiologically, first note that uniqueness of the disease-free equilibrium follows from the fact that although the isolated cholera compartment has equilibria at the quorum threshold $A_C$ and the carrying capacity $K_C$, neither can serve as a disease-free state since they generate a strictly positive force of infection in humans.
As a consequence, the only equilibrium for the cholera component that allows no infection is the one where cholera is absent.

Existence of the higher endemic equilibrium $\bE_m$ arises because the isolated bacterial population has a locally asymptotically stable equilibrium at the carrying capacity $W = K_C$, which positive human shedding shifts slightly above $K_C$ until balanced by cubic high-density bacterial crowding and bounded human prevalence.

\begin{figure}[htbp]
\centering
\begin{tikzpicture}
\begin{axis}[
    width=0.9\linewidth, height=0.55\linewidth,
    axis x line=bottom,
    axis y line=left,
    xmin=0, xmax=2.5,
    ymin=0, ymax=9,
    xlabel={$\R_v$},
    ylabel={$\{W^\star,I^\star\}$},
    xlabel style={at={(ticklabel* cs:1)}, anchor=north},
    ylabel style={at={(ticklabel* cs:1.08)}, anchor=south east, rotate=0},
    xtick={1, 1.5625},
    xticklabels={$1$, $\R_v^{SN}$},
    ytick={4, 7},
    yticklabels={$A_C$, $K_C$},
    clip=false,
    % legend cell align={left},
    % legend style={at={(0.02,0.98)}, anchor=north west, draw=black!30, fill=white, nodes={scale=0.9}}
]

% 1. DFE (stable).
\addplot[domain=0:1, samples=2, ultra thick, black] {0};
\node[anchor=south, text=black] at (axis cs:0.5, 0.05) {$\bE_0$};

% 2. DFE (unstable).
\addplot[domain=1:2.4, samples=2, ultra thick, black, dashed] {0};
\node[anchor=south, text=black] at (axis cs:2.1, 0.05) {$\bE_0$};

% ==========================================
% ENVIRONMENTAL LOAD (W*) BRANCHES
% ==========================================
% 3. Micro-endemic E_\ell (stable).
\addplot[domain=0:1.5, samples=50, ultra thick, blue, variable=\t] ({1 + 0.75*\t - 0.25*\t^2}, {\t});

% 4. Unstable Allee threshold E_u (unstable saddle).
\addplot[domain=1.5:4, samples=50, ultra thick, blue, dashed, variable=\t] ({1 + 0.75*\t - 0.25*\t^2}, {\t});

% 5. Macro-endemic E_m (stable).
\addplot[domain=0:2.4, samples=50, ultra thick, blue] {7 + 0.5*x - 0.05*x^2};

% ==========================================
% HUMAN PREVALENCE (I*) BRANCHES
% ==========================================
% E_\ell (stable)
\addplot[domain=0:1.5, samples=50, ultra thick, teal, variable=\t] ({1 + 0.75*\t - 0.25*\t^2}, {0.8*\t*(4-\t)});

% E_u (unstable)
\addplot[domain=1.5:4, samples=50, ultra thick, teal, dashed, variable=\t] ({1 + 0.75*\t - 0.25*\t^2}, {0.8*\t*(4-\t)});

% E_m (stable)
\addplot[domain=0:2.4, samples=50, ultra thick, teal] {3 + x};

% --- Annotations and markers ---

% Transcritical bifurcation dot.
\filldraw[black] (axis cs:1,0) circle (2.5pt);

% Saddle-node bifurcation dot.
\filldraw[black] (axis cs:1.5625, 1.5) circle (2.5pt);  % W* turn
\filldraw[black] (axis cs:1.5625, 3.0) circle (2.5pt);  % I* turn
\draw[dotted, thick, gray] (axis cs:1.5625, 3.0) -- (axis cs:1.5625, 0);

% Labelling the branches.
\node[anchor=north west, text=blue] at (axis cs:1.5, 1.05) {$\bE_\ell (W^\star)$};
\node[anchor=south east, text=blue] at (axis cs:0.4, 3.9) {$\bE_u (W^\star)$};
\node[anchor=south, text=blue] at (axis cs:1.8, 7.7) {$\bE_m (W^\star)$};

\node[anchor=north west, text=teal] at (axis cs:1.5, 2.65) {$\bE_\ell (I^\star)$};
\node[anchor=south east, text=teal] at (axis cs:0.6, 1.4) {$\bE_u (I^\star)$};
\node[anchor=south, text=teal] at (axis cs:2.0, 5.1) {$\bE_m (I^\star)$};
\end{axis}
\end{tikzpicture}
\caption{Schematic bifurcation diagram depicting the equilibrium environmental pathogen load $W^\star$ (blue) and human prevalence $I^\star$ (teal) as functions of the vaccinated reproduction number $\R_v$.
The disease-free equilibrium $\bE_0$ exists for both variables and is shown in black.
Continuous curves are locally asymptotically stable, dashed ones are unstable.
$\R_v^{SN}$ denotes the saddle-node bifurcation point.
}
\label{fig:forward-bifurcation}
\end{figure}

More generally, considering Figure~\ref{fig:forward-bifurcation} we see that the situation is very different from the typical bistability that arises in epidemiological models (e.g., \cite{arino2003global,hadeler1997backward,Kribs-ZaletaVelasco-Hernandez2000}).
In classic models with bistability, the saddle-node bifurcation occurs at some value $\R_v^{SN}<1$ and the infection remains eliminable whenever $\R_v<\R_v^{SN}$, with the disease-free equilibrium being typically globally asymptotically stable there.
It is instructive to compare this picture with the results of \cite{Kolaye_Guilsou_2023}, who studied a 4-dimensional $SIRS\text{--}B$ cholera model with waterborne transmission and a strong cubic Allee effect (without vaccination, direct transmission or infectious corpses) and used a cubic polynomial reduction together with a centre-manifold calculation at the epidemic threshold to link the Allee effect to backward bifurcation and bistability.
In our 6-dimensional $SIRVDW$ model, which extends that framework to include imperfect waning vaccination and infectious corpses, our geometric nullcline analysis and the centre-manifold analysis in Appendix~\ref{app:vaccine-backward-bifurcation} disentangle the local bifurcation at $\bE_0$ ($\R_v = 1$) from the global Allee-driven bistability.
Specifically, whenever the Allee threshold exceeds the half-saturation concentration ($A_C > K_C/(K_C - 1) \approx 1$), the environmental Allee and saturation terms actually oppose a local backward bifurcation at $\bE_0$ (Remark~\ref{rem:bifurcation-preclusion}), so the transcritical bifurcation at $\R_v = 1$ is forward ($a < 0$), while global bistability persists for all $\R_v < \R_v^{SN}$ because the micro-endemic branch $\bE_\ell$ folds backward at $\R_v = \R_v^{SN} > 1$ into the unstable Allee branch $\bE_u$ alongside the ever-present macro-endemic state $\bE_m$ ($W^\star > K_C$).
If a sudden influx of contamination pushes the concentration of cholera past the unstable branch $\bE_u(W^\star)$, the bacteria autonomously replicate toward $K_C$. 
Once this higher endemic equilibrium is established, vaccination alone ($\R_v < 1$) cannot eliminate the disease; while hygiene and sanitation measures that reduce human shedding ($\zeta$) and exposure ($\beta_W$) can shield the human population and prevent an initial outbreak from crossing the Allee threshold, collapsing an already established environmental reservoir requires water treatment interventions (such as chlorination, increasing $r_C$ or directly reducing $W$) to force the bacterial concentration back below $\bE_u(W^\star)$.
We return to this in Section~\ref{subsec:WASH-vs-vaccination}, where our results underscore the need for multi-pronged cholera control combining vaccination with WASH interventions \cite{azman2020surveillance}: unlike smallpox (for global eradication) or measles (for regional elimination), no single intervention mechanism is sufficient on its own to achieve cholera elimination, both at the local scale of an outbreak and more generally (there is no ``silver bullet'').

%%%%%%%%%%%%%%%%%%%%%%%%%%%%%%%
%%%%%%%%%%%%%%%%%%%%%%%%%%%%%%%
%%%%%%%%%%%%%%%%%%%%%%%%%%%%%%%
%%%%%%%%%%%%%%%%%%%%%%%%%%%%%%%
\section{Computational considerations} 
\label{sec:computational-analysis}

\begin{table}[htbp]
\centering
\footnotesize
\begin{tabular}{lllll}
\toprule
Parameter & Plausible range & Default value & Unit & Source \\
\midrule
\multicolumn{5}{c}{Human-related parameters} \\
$b$ & -- & -- & people$\times$day$^{-1}$ & Computed \\ 
$1/d_H$  & -- & $52.5 \times 365$ & day & \cite{worldbank_life_expectancy} \\ 
$1/e$ & $[1, 15]$ & 2 & day & \cite{gunnlaugsson1998funerals,ifrc2018safe} \\
$1/\delta$ & $[1, 50]$ & 7 & day & \cite{canada_cholera}\\ 
$1/\gamma$ & $[2, 60]$ & 10 & day & \cite{pasteur_cholera}\\ 
$1/\varepsilon$ & $[365, 1825]$ & 730 & day & \cite{pasteur_cholera}\\
$1/v$ & $[10, 730]$ & 100 & day & Assumed\\
$\sigma$ & $(0, 1)$ & 0.7 & dimensionless & \cite{who_cholera_vaccination} \\
$1/\theta$ & $[365, 1825]$ & 730 & day & \cite{who_cholera_vaccination}\\ 
\midrule
\multicolumn{5}{c}{\emph{V. cholerae}-related parameters} \\
$\zeta$ & $[0.01, 2]$ & 1.043 & people$^{-1}\times$day$^{-1}$& Calibrated (\cite{canada_cholera}) \\ 
$\eta$ & $[0, 1]$ & 1 & dimensionless & Assumed \\
$1/r_C$ & $[0.5, 30]$ & 2 & day & Assumed \\
$A_C$ & $[10^3, 10^5]$ & $10^4$ & dimensionless & Assumed \\
$K_C$ & $[10^5, 10^7]$ & $10^6$ & dimensionless & Assumed \\
\midrule
\multicolumn{5}{c}{Disease transmission-related parameters} \\
$\R_0$ & $[1.1, 20]$ & 10 & dimensionless & See text \\
$\beta_{I}$ & $[0.01, 1.0]$ & 0.1 & day$^{-1}$ & Assumed\\ 
$\beta_{D}$ & $[0.01, 1.0]$ & 0.1 & day$^{-1}$& Assumed \\ 
$\beta_{W}$ & $[10^{-6}, 2\times 10^{-5}]$ & $8.58\times 10^{-6}$ & day$^{-1}$& Calibrated \\ 
\bottomrule	
\end{tabular}
\caption{Baseline parameter values and plausible ranges used in the numerical investigations (scenario-specific modifications are stated in the text); see Table~\ref{tab:model-parameters} for parameter definitions.} 
\label{tab:parameter-values_v}
\end{table}

For parameters, we consider ranges given in Table~\ref{tab:parameter-values_v}.
Let us first comment on some of the estimated parameters.
First, we compute $b$ so that in the absence of disease, $b/d_H$ gives the population for the location under consideration. 
This is taken to be 100,000 individuals.

We make use of the value of the basic reproduction number $\R_0$ to determine the baseline value of the waterborne transmission parameter $\beta_W$ given $\beta_I$ and $\beta_D$.
The value of $\R_0$ varies depending on several factors, in particular the geographical context, environmental conditions, access to potable water and population density.
In the literature, reported and modelled values of the basic reproduction number $\R_0$ range from approximately 1.3 \cite{lupica2020computation} and 2.8 \cite{Yang_2018} to 4--9 \cite{lupica2020computation,sisodiya2018pathogen} and up to 18 in historical epidemics or when accounting for hyperinfective transmission \cite{chan2013historical}.
As a consequence, we use the range $[1.1, 20]$ for variations of $\R_0$, with a default baseline value of $\R_0 = 10$ (corresponding to $\R_v \simeq 3.87$ at baseline vaccination parameters).

Note that our model does not incorporate treatment \emph{explicitly}. As a consequence, the recovery rate $\gamma$ and disease-induced death rate $\delta$ reflect a composite of untreated and treated cases: while severe untreated cholera can be fatal within 1--3 days (with case fatality ratios of 25\% to 50\%), our baseline values ($1/\delta = 7$ days and $1/\gamma = 10$ days) represent an under-treated outbreak setting with high mortality and we consider wide ranges ($[1, 50]$ days for $1/\delta$ and $[2, 60]$ days for $1/\gamma$) to span rapid untreated progression through prolonged community recovery.

%%%%%%%%%%%%%%%%%%%%%%%%%%%%
%%%%%%%%%%%%%%%%%%%%%%%%%%%%
\subsection{Effect of the parameters}
\begin{figure}[htbp]
    \centering
    \begin{subfigure}{0.493\textwidth}
        \centering
        \caption{PRCC of $\R_0$}
        \includegraphics[width=\textwidth]{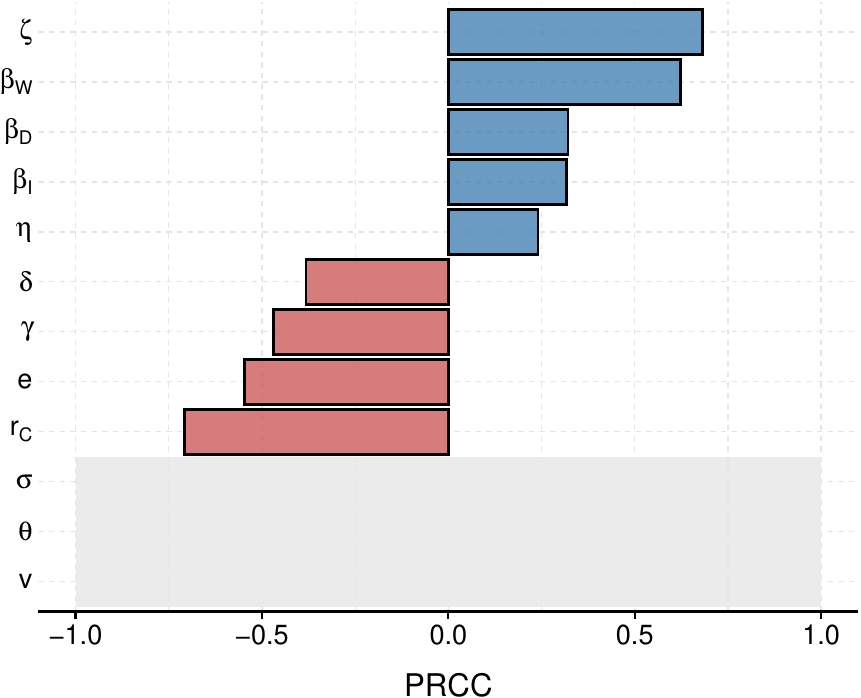}
        \label{fig:PRCC-R0}
    \end{subfigure}\hfill
    \begin{subfigure}{0.493\textwidth}
        \centering
        \caption{PRCC of $\R_v$}
        \includegraphics[width=\textwidth]{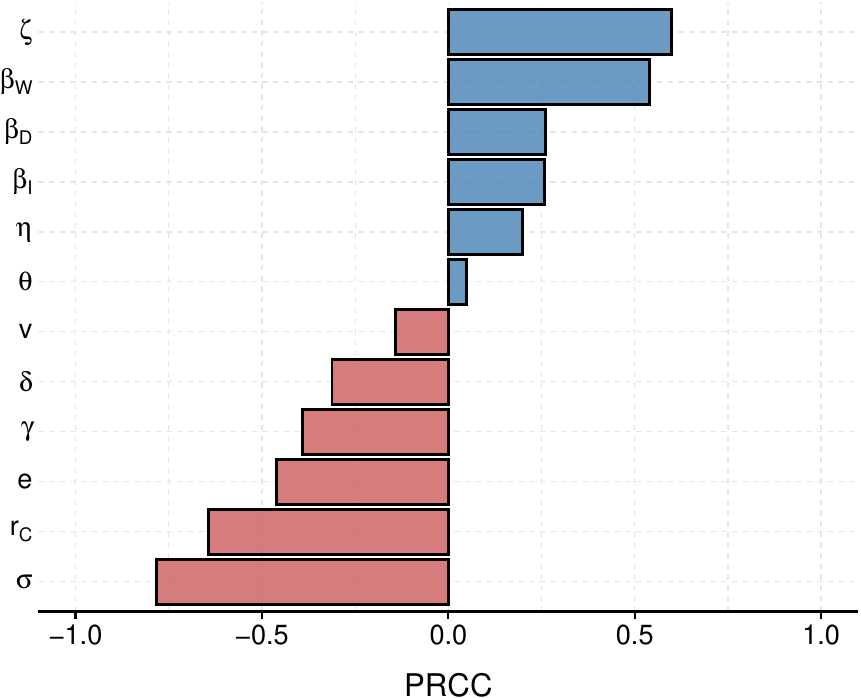}
        \label{fig:PRCC-Rv}
    \end{subfigure}
    \caption{Partial rank correlation coefficients (PRCC) of the basic reproduction number $\R_0$ (a) and the vaccinated reproduction number $\R_v$ (b), computed from $100{,}000$ Latin Hypercube samples across the parameter ranges in Table~\ref{tab:parameter-values_v}.}
    \label{fig:PRCC}
\end{figure}

We first perform a global sensitivity analysis of $\R_0$ (Figure~\ref{fig:PRCC-R0}) and $\R_v$ (Figure~\ref{fig:PRCC-Rv}) with respect to the model parameters, via a global sensitivity analysis using partial rank correlation coefficients (PRCC) computed over $100{,}000$ Latin Hypercube samples.
Given the wide biological uncertainty in our parameter estimates, we utilise a global sensitivity analysis across the parameter space rather than local analytical partial derivatives, which allows us to robustly capture parameter interactions across their entire plausible ranges.
Results are shown in Figure~\ref{fig:PRCC}.
Both $\R_0$ and $\R_v$ exhibit the same sensitivity profile with respect to the non-vaccine parameters: the strongest positive drivers are the rate $\zeta$ at which humans contaminate the water and the waterborne transmission rate $\beta_W$, while the strongest negative drivers are the low-density bacterial decay rate $r_C$ and (for $\R_v$) the vaccine efficacy $\sigma$.
The recovery rate $\gamma$ and the disease-induced death rate $\delta$ also reduce $\R_0$ and $\R_v$ by shortening the infectious period, although $\delta$ is of course not a control parameter.
Besides vaccine efficacy $\sigma$, the vaccination rate $v$ also plays a significant role in lowering $\R_v$.
In addition to the moderate positive sensitivity to direct person-to-person transmission ($\beta_I$, $\text{PRCC} \simeq 0.26\text{--}0.32$) and direct corpse-to-human transmission ($\beta_D$, $\text{PRCC} \simeq 0.26\text{--}0.32$), the interment rate $e$ ($\text{PRCC} \simeq -0.46\text{ to }-0.55$) and the relative post-mortem shedding factor $\eta$ ($\text{PRCC} \simeq 0.20\text{--}0.24$) have a substantial influence on both $\R_0$ and $\R_v$ because unburied corpses contribute both to direct funeral transmission and to environmental water contamination through the effective shedding rate $\zeta_{\text{eff}} = \zeta(1 + \eta\delta/e)$.

%%%%%%%%%%%%
\begin{figure}[htbp]
    \centering
    \begin{subfigure}{0.493\textwidth}
        \centering
        \includegraphics[width=\textwidth]{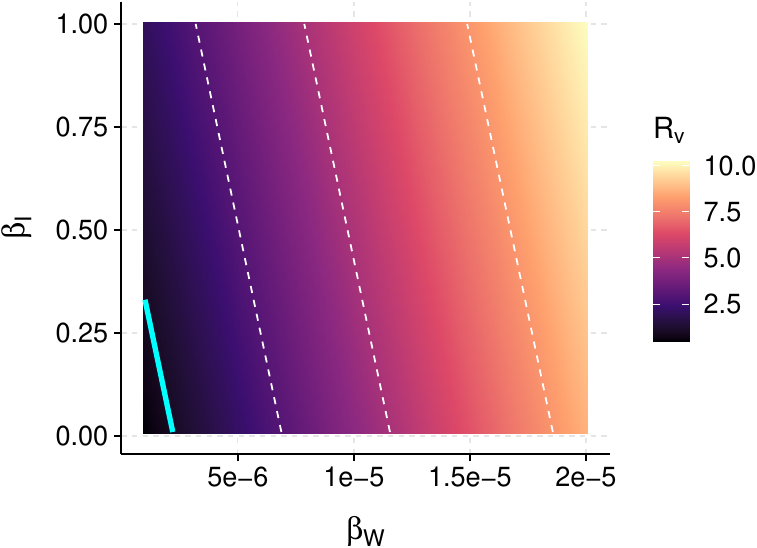}
        \caption{$\beta_W$ and $\beta_I$}
        \label{fig:Rv-fct-betaW-betaI}
    \end{subfigure}
    \begin{subfigure}{0.493\textwidth}
        \centering
        \includegraphics[width=\textwidth]{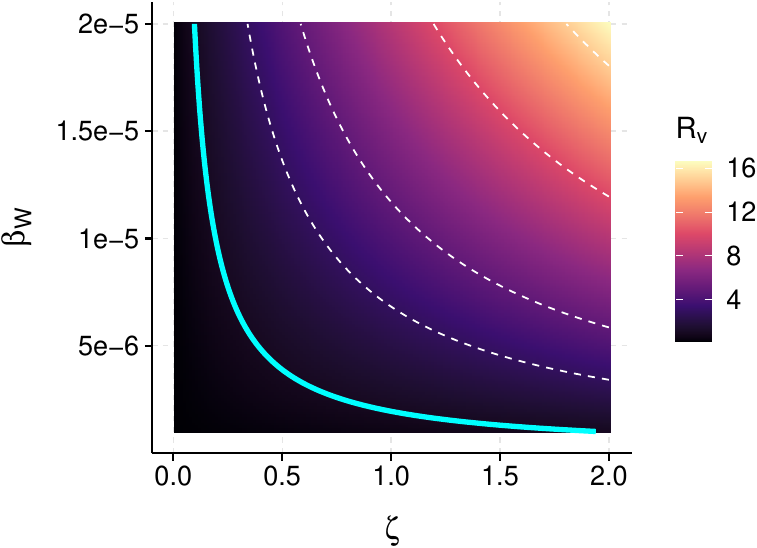}
        \caption{$\zeta$ and $\beta_W$}
        \label{fig:Rv-fct-zeta-betaW}
    \end{subfigure} \\
    \begin{subfigure}{0.493\textwidth}
        \centering
        \includegraphics[width=\textwidth]{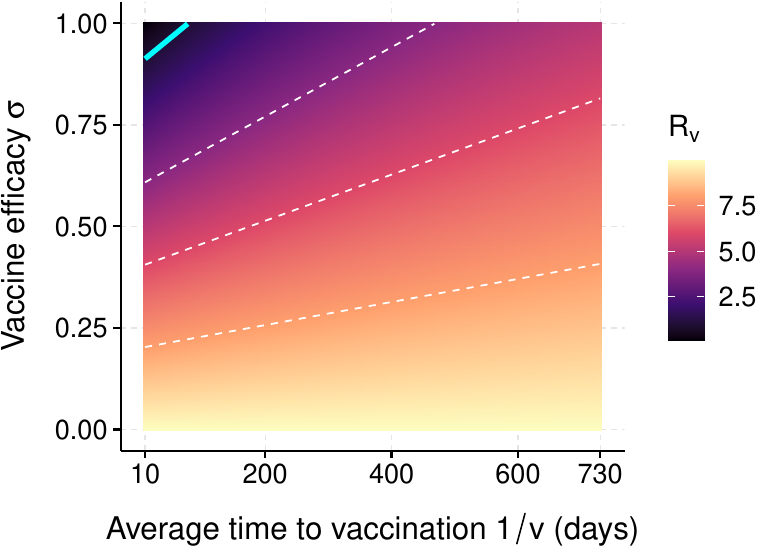}
        \caption{$1/v$ and $\sigma$}
        \label{fig:Rv-fct-sigma-v}
    \end{subfigure}
    \begin{subfigure}{0.493\textwidth}
        \centering
        \includegraphics[width=\textwidth]{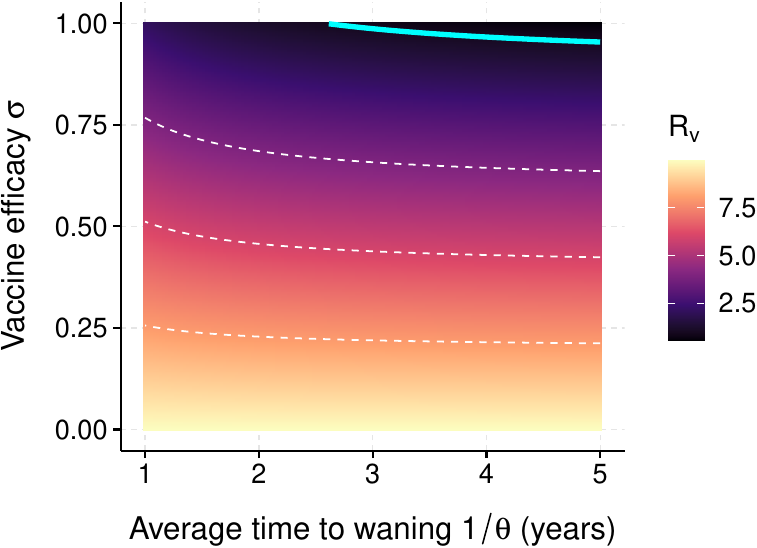}
        \caption{$1/\theta$ and $\sigma$}
        \label{fig:Rv-fct-theta-sigma}
    \end{subfigure}    
    \caption{Variation of $\R_v$ as a function of (a) water-to-human ($\beta_W \in [10^{-6}, 2\times 10^{-5}]$) and human-to-human ($\beta_I \in [0.01, 1.0]$) transmission coefficients, (b) pathogen shedding rate ($\zeta \in [0.01, 2]$) and water-to-human transmission coefficient ($\beta_W \in [10^{-6}, 2\times 10^{-5}]$), (c) mean time to vaccination ($1/v \in [10, 730]$ days) and vaccine efficacy ($\sigma \in (0, 1)$) and (d) mean time to vaccine waning ($1/\theta \in [1, 5]$ years) and vaccine efficacy ($\sigma \in (0, 1)$).
    The solid cyan curve marks the threshold $\R_v = 1$ and the dashed white curves show contour lines of $\R_v$.
    All other parameters are fixed at their default values in Table~\ref{tab:parameter-values_v}.}
    \label{fig:heatmaps-Rv}
\end{figure}

In the heatmaps of Figure~\ref{fig:heatmaps-Rv}, we explore in more detail how different parameters influence the vaccinated reproduction number $\R_v$.
First, observe in Figure~\ref{fig:Rv-fct-betaW-betaI} that across the full plausible range of $\beta_I \in [0.01, 1.0]$, $\R_v$ depends much less on $\beta_I$ (human-to-human transmission) than it does on $\beta_W$ (water-to-human transmission).
The situation is the same when plotting $\R_v$ as a function of $\beta_W$ and $\beta_D$ or $\beta_I$ and $\beta_D$, although we do not show these plots.

Figure~\ref{fig:Rv-fct-zeta-betaW} then considers the effect of two parameters whose values can be modified by WASH interventions: the rate $\zeta$ at which humans contaminate water and the coefficient $\beta_W$ of water-to-human pathogen transmission.
Here, the effect is similar: decreasing the rate of contamination of water by humans, e.g., by using proper sanitation methods, or that of contamination of humans by contaminated water, e.g., using hand washing or filtration, has the effect of greatly reducing the vaccinated reproduction number $\R_v$.

We now turn to the effect of vaccination.
In Figure~\ref{fig:Rv-fct-sigma-v}, we consider the average time to vaccination in days and the vaccine efficacy $\sigma$.
Note that in an epidemic situation, high waiting times for vaccination as shown towards the right of Figure~\ref{fig:Rv-fct-sigma-v} are unrealistic, but they may be relevant in a routine vaccination scenario.
The ideal situation here is, unsurprisingly, with a short waiting time to vaccination and high vaccine efficacy.
The latter is a characteristic of the vaccine and hard to modify in the short term, whereas the former is directly accessible through public health policy.
Likewise, Figure~\ref{fig:Rv-fct-theta-sigma} shows how $\R_v$ changes as a function of the mean time $1/\theta$ to loss of vaccine protection and vaccine efficacy $\sigma$.
We see that vaccine efficacy is the main driver of $\R_v$ here, with a marked increase in $\R_v$ when the mean time to waning is short. 

%%%%%%%%%%%%%%%
%%%%%%%%%%%%%%%
\subsection{Impact of disease-induced mortality and time-to-interment}
\label{subsec:disease-induced-death-and-interment}

The risks of cholera transmission associated with dead bodies are mostly due to the fact that after death, the muscles controlling sphincters relax, so the bacteria contained in the digestive system of an individual are released into clothing and the environment \cite{gunnlaugsson1998corpses}.
Regions with high cholera presence have a variety of post-death practices, both in terms of pre-interment rites and time-to-interment; see Table~\ref{tab:cholera-risks} for a short overview.

\begin{table}[htbp]
\centering
\resizebox{\textwidth}{!}{%
\begin{tabular}{@{}llll@{}}
\toprule
\textbf{Practice/tradition} & \textbf{Typical window} & \textbf{Primary transmission pathway} & \textbf{Evidence and references} \\ 
\midrule
Rapid burial and ritual washing & $<$ 24 hours (Islamic, Balanta) & Direct contact with bodily fluids during washing & \cite{gunnlaugsson1998funerals,suwalowska2021ethical} \\
Prolonged wakes and shrouding   & 2--5 days (Papel/ATR); days--weeks (Christian) & Direct contact (touching, washing, wrapping corpse) & $\text{OR} = 13.4$ for touching corpse \cite{gunnlaugsson1998funerals,suwalowska2021ethical} \\
Communal funeral meals          & During funerals         & Corpse handlers preparing/serving food; shared bowls & $\text{OR} = 10.4$ for eating at funeral \cite{gunnlaugsson1998funerals,ngwa2017cultural} \\
Undisinfected home preparation  & 1--3 days (mean 43 h)   & Transporting, washing and wrapping undisinfected body & $2.6\times$ village attack rate \cite{gunnlaugsson1998corpses,gunnlaugsson1998funerals} \\
Shallow burial and runoff       & Rainy season / long-term & Runoff from shallow graves into wells; necroleachate & \cite{neckel2017environmental,ngwa2017cultural} \\ 
\bottomrule
\end{tabular}%
}
\caption{Summary of cholera transmission pathways across funeral and burial practices reported in the literature. ATR: African traditional religions; OR: odds ratio.}
\label{tab:cholera-risks}
\end{table}

To better understand the role of deceased individuals, we consider the effect of interment practices and disease-induced mortality.
In Figure~\ref{fig:heatmap-e-delta}, we vary the mean time to interment $1/e$ and the mean time to death by the disease $1/\delta$, with an initial condition of $10$ infected individuals ($I(0) = 10$, with $S(0)$ and $V(0)$ at their disease-free equilibrium proportions and $R(0) = D(0) = W(0) = 0$). 
We then compute numerically the area under the curves $\beta_D DS/N_H + (1-\sigma) \beta_D DV/N_H$ and $\delta I$ over $365$ days, respectively, giving the total number of infections directly due to contact with dead bodies (Figure~\ref{fig:heatmap-e-delta-incidence}) and the overall number of deaths due to cholera (Figure~\ref{fig:heatmap-e-delta-deaths}).

\begin{figure}[htbp]
  \centering
  \begin{subfigure}[b]{0.48\textwidth}
    \centering
    \includegraphics[width=\textwidth]{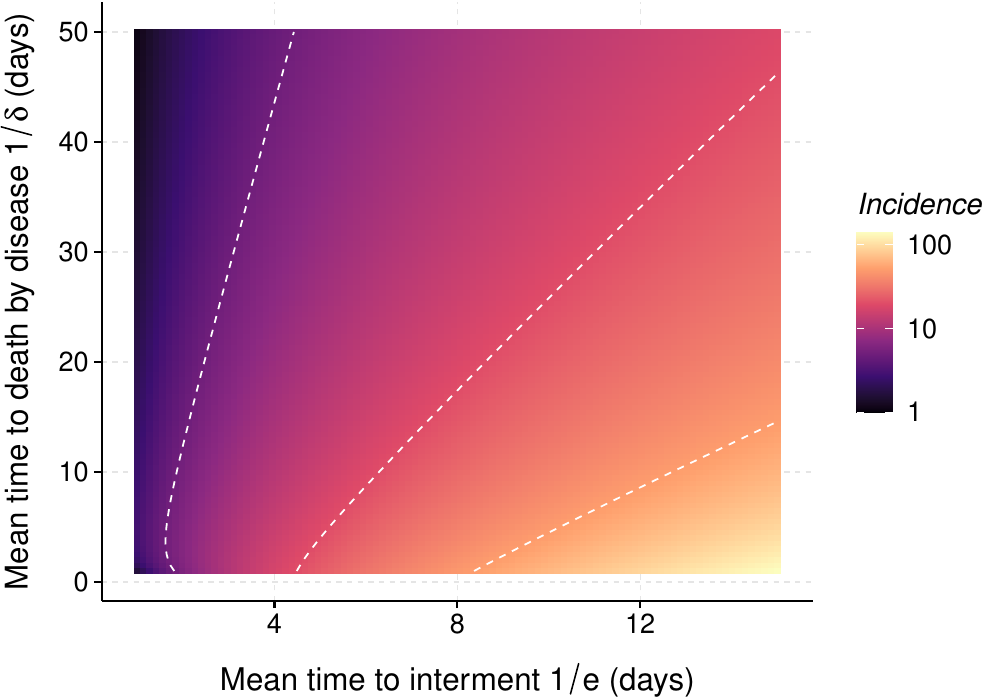}
    \caption{Incidence due to dead bodies}
    \label{fig:heatmap-e-delta-incidence}
  \end{subfigure}
  \hfill
  \begin{subfigure}[b]{0.48\textwidth}
    \centering
    \includegraphics[width=\textwidth]{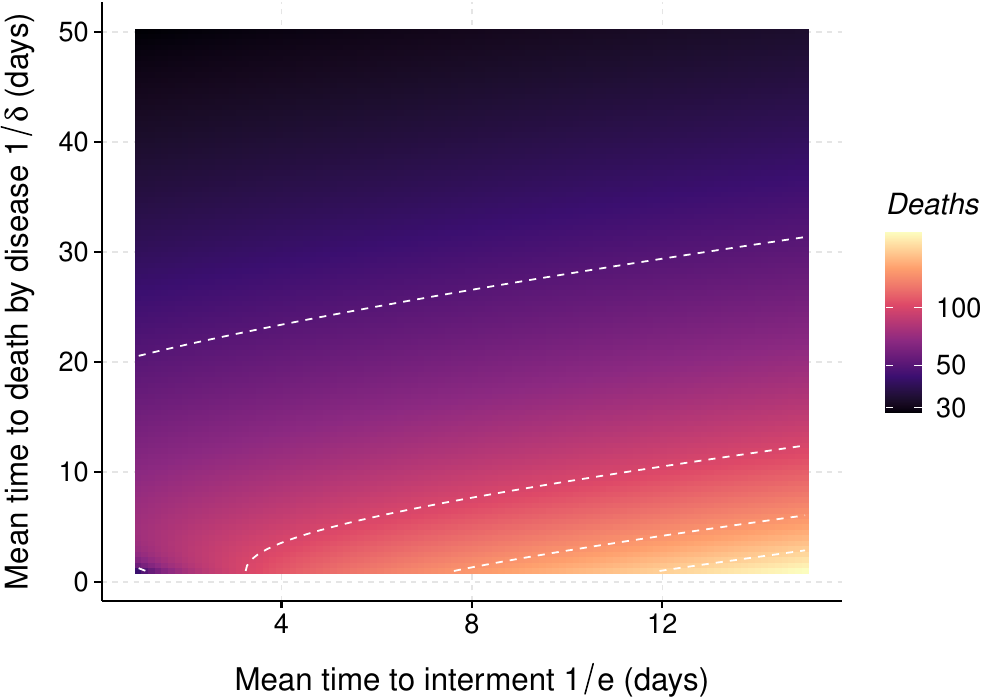}
    \caption{Disease-induced deaths}
    \label{fig:heatmap-e-delta-deaths}
  \end{subfigure}
  \caption{Total 1-year incidence of contaminations due to contacts with dead bodies (a) and disease-induced deaths (b) as a function of the mean time to interment $1/e$ and the mean time to death by the disease $1/\delta$.}
  \label{fig:heatmap-e-delta}
\end{figure}

Figure~\ref{fig:heatmap-e-delta-incidence} confirms that longer delays before interment ($1/e$) and shorter times to disease-induced death ($1/\delta$) increase the cumulative number of infections resulting from contacts with deceased individuals.
Simultaneously, Figure~\ref{fig:heatmap-e-delta-deaths} shows that both rapid disease-induced mortality and delayed interment increase overall cholera deaths over the one-year period, since unburied corpses both infect funeral attendees directly and shed bacteria into the water supply.

\subsection{Safe and dignified burials}
To mitigate the transmission risks associated with handling corpses, health organisations have established strict guidelines for ``safe and dignified burials'' (SDB) \cite{ifrc2018safe,who2017safe}. 
These guidelines emphasise limiting physical contact with the deceased, disinfecting the body and its immediate environment and replacing traditional communal washing rites with safe, culturally adapted alternatives guided by trained burial teams. 
Adherence to these protocols effectively reduces the transmission rate from dead bodies $\beta_D$ and shortens the time to interment $1/e$. 

In Figure~\ref{fig:heatmap-safe-burial}, we explore the epidemiological benefits of adopting SDB practices in an outbreak regime near the epidemic threshold.
Specifically, we scale $\beta_W$ to $14\%$ of its baseline value in Table~\ref{tab:parameter-values_v} so that $\R_0\simeq 2$ when $e$ and all other parameters are at their default values.
We then vary the mean time to interment $1/e$ and the percentage reduction in $\beta_D$ and compute total incidence and disease-induced deaths over a one-year period from the same initial condition of $10$ infected individuals as in Figure~\ref{fig:heatmap-e-delta}.

\begin{figure}[htbp]
  \centering
  \begin{subfigure}[b]{0.48\textwidth}
    \centering
    \includegraphics[width=\textwidth]{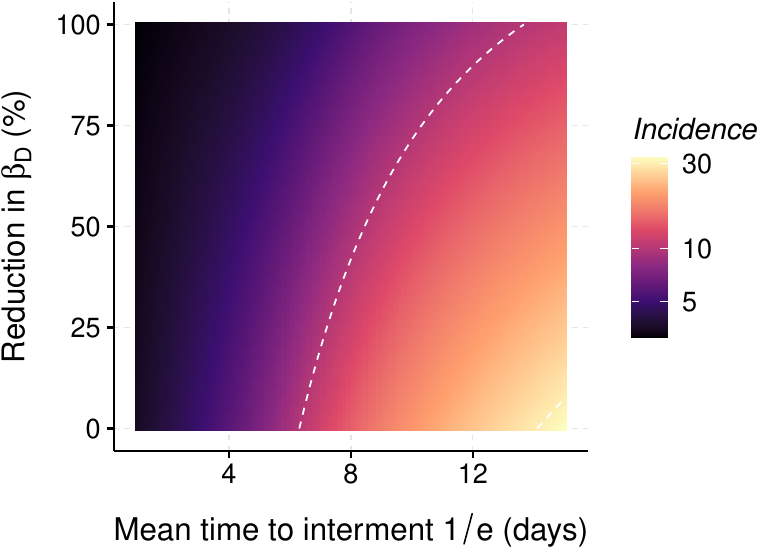}
    \caption{Total incidence}
    \label{fig:heatmap-safe-burial-incidence-all}
  \end{subfigure}
  \begin{subfigure}[b]{0.48\textwidth}
    \centering
    \includegraphics[width=\textwidth]{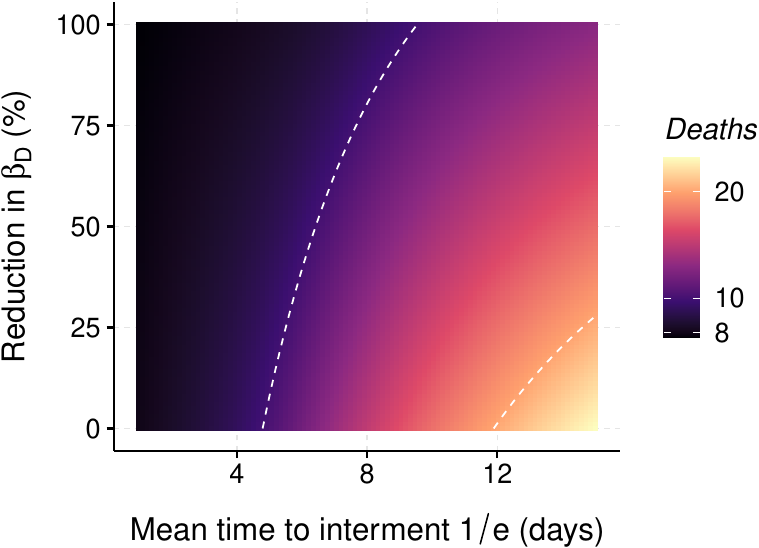}
    \caption{Disease-induced deaths}
    \label{fig:heatmap-safe-burial-deaths}
  \end{subfigure}
  \caption{Epidemiological impact of safe and dignified burials. 
  The total incidence from all sources (\subref{fig:heatmap-safe-burial-incidence-all}) and total disease-induced deaths (\subref{fig:heatmap-safe-burial-deaths}) are plotted as functions of the mean time to interment $1/e$ and the percentage reduction in the transmission parameter $\beta_D$.}
  \label{fig:heatmap-safe-burial}
\end{figure}

Figure~\ref{fig:heatmap-safe-burial} illustrates that combining rapid interment with a substantial reduction in direct corpse transmission ($\beta_D$) limits both total cases and cholera deaths, highlighting the value of safe burial protocols in cholera outbreak management.

%%%%%%%%%%%%%%%%%%%%%%
%%%%%%%%%%%%%%%%%%%%%%
\subsection{Comparing the nature of introductions}

To better understand the role of infection stemming from improper burials, we consider pathogen introductions by way of a single infectious source at time $t=0$ inside a small hamlet population of $N_H = 100$ individuals ($b/d_H = 100$).
In low-resource settings, oral cholera vaccine (OCV) campaigns are typically deployed in localised high-risk hotspots, achieving moderate coverage.
Here, we simulate a targeted vaccination campaign with a rate of $v = 0.0012\text{ day}^{-1}$ (mean time to vaccination $\simeq 2.3$ years) and vaccine efficacy of $\sigma = 0.7$, which yields a disease-free vaccine coverage of $\Psi_0^v \simeq 45.8\%$ ($S_0 \simeq 54.2, V_0 \simeq 45.8$).
To represent a small local water source, the hamlet scenario uses $\beta_I = \beta_D = 0.05\text{ day}^{-1}$, $\beta_W = 0.04\text{ day}^{-1}$, $\zeta = 0.078\text{ person}^{-1}\text{day}^{-1}$, $r_C = 2.0\text{ day}^{-1}$, $A_C = 0.15$ and $K_C = 5.0$ (with all other rates as in Table~\ref{tab:parameter-values_v}).
Under these parameters, $\R_0 \simeq 1.09 > 1$ in the absence of vaccination, while the targeted OCV campaign reduces the vaccinated reproduction number below unity ($\R_v \simeq 0.74 < 1$), rendering the disease-free equilibrium locally asymptotically stable.
Starting from this vaccinated disease-free demographic baseline at $t=0$, we compare the consequences of introducing either a single active human case ($I(0) = 1$, $S(0) = S_0 - 1$, $V(0) = V_0$) or a single unburied corpse ($D(0) = 1$, $S(0) = S_0$, $V(0) = V_0$).

\begin{figure}[htbp]
  \centering
  \begin{subfigure}{0.493\textwidth}
    \centering
    \includegraphics[width=\textwidth]{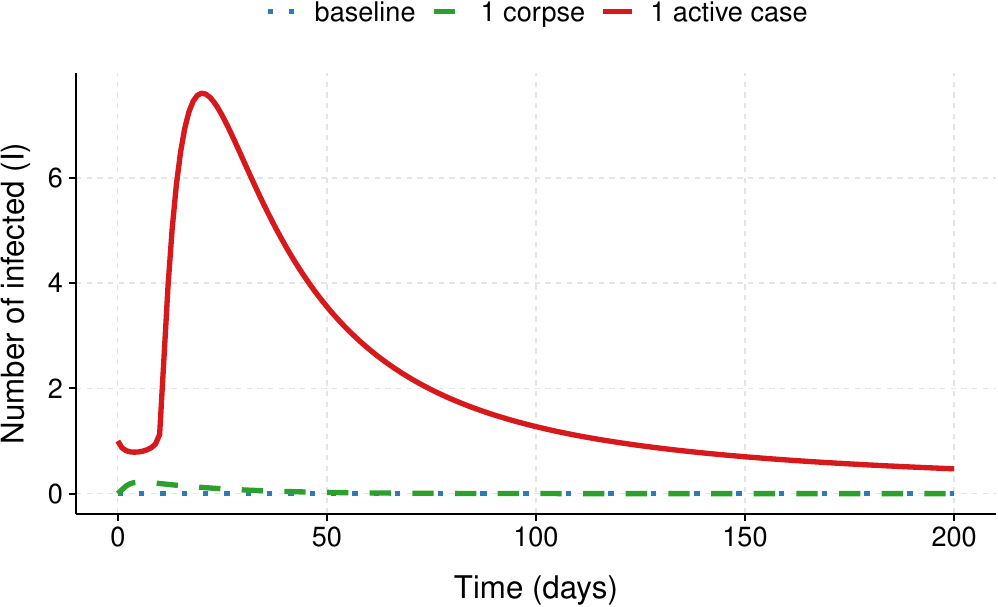}
    \caption{Infected human prevalence $I(t)$}
    \label{fig:plot-funeral-impulse-I}
  \end{subfigure}
  \begin{subfigure}{0.493\textwidth}
    \centering
    \includegraphics[width=\textwidth]{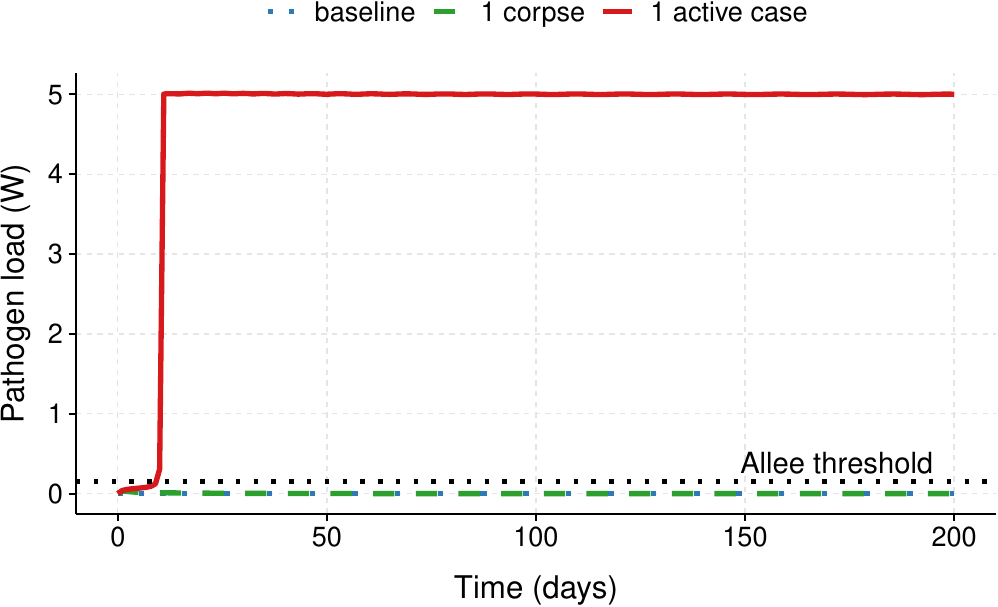}
    \caption{Environmental pathogen load $W(t)$}
    \label{fig:plot-funeral-impulse-W}
  \end{subfigure}
  \caption{Comparison of the nature of the introduced case (single active case versus single unburied corpse) at time $t=0$ in a hamlet of $N_H = 100$ individuals under targeted vaccination ($\R_v \simeq 0.74 < 1$ with $\simeq 46\%$ coverage). 
  (a) Number of infected humans $I(t)$ (equivalently, percentage prevalence in the hamlet). 
  (b) Dimensionless environmental pathogen load $W(t)$, with the horizontal dotted line indicating the Allee threshold $A_C = 0.15$.}
  \label{fig:plot-funeral-impulse}
\end{figure}

Figure~\ref{fig:plot-funeral-impulse} compares these introduction scenarios against a pathogen-free baseline. While the baseline stays at zero and the unburied corpse scenario decays back to zero (because the corpse is buried after 2 days on average, before its shedding can breach the Allee threshold $A_C = 0.15$ in the partially vaccinated population), the introduction of a single active human case triggers a full-blown epidemic. 
This pushes the environmental pathogen load $W(t)$ past the Allee threshold, forcing the environment into its stable carrying capacity state ($W \to K_C = 5.0$) and driving a human epidemic wave that peaks at a prevalence of $7.6\%$ of the hamlet population.
After the acute epidemic wave subsides, human prevalence settles to a low endemic level sustained by the environmental reservoir $W^\star \approx K_C$.

%%%%%%%%%%%%%%%%%%%%%%
%%%%%%%%%%%%%%%%%%%%%%
\subsection{Trade-off between WASH and vaccination}
\label{subsec:WASH-vs-vaccination}
As noted earlier, no single intervention is a panacea for cholera; in an ideal setting, all control mechanisms would be mobilised concurrently in order to curtail an outbreak.
However, in practice, resources are often limited, making it valuable to evaluate the trade-offs between interventions.

\begin{figure}[htbp]
  \centering
  \begin{subfigure}{0.493\textwidth}
    \centering
    \includegraphics[width=\textwidth]{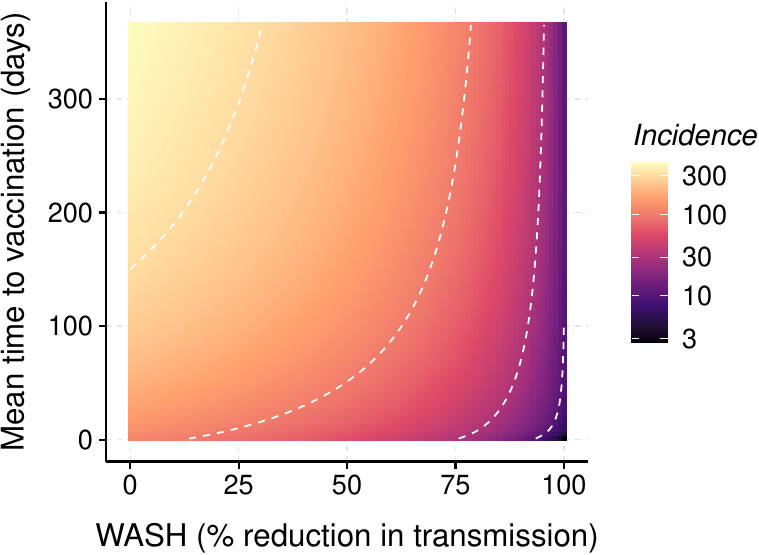}
    \caption{WASH vs. mean time to vaccination ($1/v$)}
    \label{fig:heatmap-tradeoff-wash-vaccination-rate}
  \end{subfigure}
  \begin{subfigure}{0.493\textwidth}
    \centering
    \includegraphics[width=\textwidth]{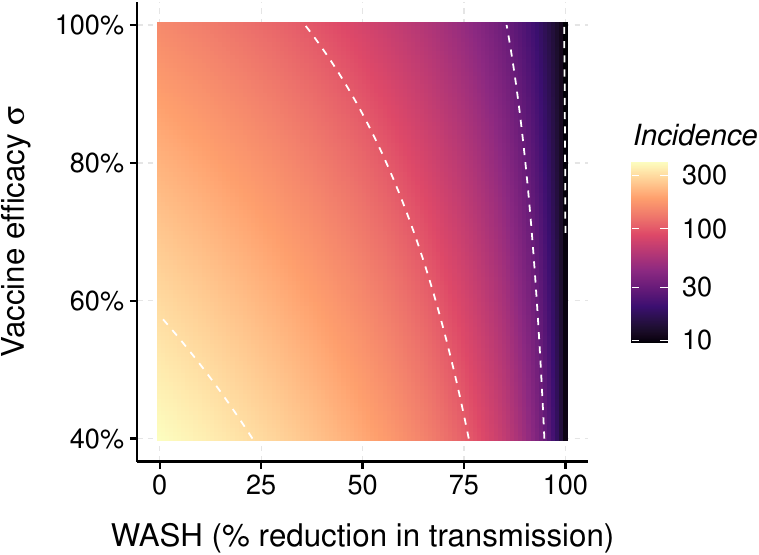}
    \caption{WASH vs. vaccine efficacy ($\sigma$)}
    \label{fig:heatmap-tradeoff-wash-vaccine-efficiency}
  \end{subfigure}
  \caption{Trade-off analysis between WASH interventions and vaccination effort. The heatmaps display the total 1-year incidence (cumulative cases, log-scaled) following an environmental contamination event ($W(0) = 1.5 A_C > A_C$), plotted as a function of the percentage reduction in water-related transmission due to WASH and (a) the mean time to vaccination $1/v$ (days) or (b) the vaccine efficacy $\sigma$.}
  \label{fig:heatmap-tradeoff-wash-vaccine}
\end{figure}

To illustrate, let us perform a trade-off analysis between WASH interventions and vaccination effort. 
Figure~\ref{fig:heatmap-tradeoff-wash-vaccine} maps the total 1-year cumulative cases under varying levels of combined interventions. 
We simulate a scenario where the environment is severely contaminated at the onset of the outbreak ($I(0) = 10$ and $W(0) = 1.5 A_C = 1.5 \times 10^4$, seeded above the Allee threshold $A_C$). 
The horizontal axis represents the strength of the WASH intervention (implemented as a percentage reduction applied uniformly to both the water-to-human transmission $\beta_W$ and human-to-water shedding $\zeta$), while the vertical axis represents either the ``agility'' of the vaccination campaign (the mean time to vaccination $1/v$, ranging from 1 to 365 days, Figure~\ref{fig:heatmap-tradeoff-wash-vaccination-rate}) or the vaccine efficacy $\sigma$ (ranging from $40\%$ to $100\%$, Figure~\ref{fig:heatmap-tradeoff-wash-vaccine-efficiency}). 
To evaluate a purely reactive vaccination strategy, we assume that no individuals are vaccinated at the onset of the outbreak ($V(0) = 0$). 

We observe in Figure~\ref{fig:heatmap-tradeoff-wash-vaccine} that WASH interventions have a more pronounced effect than vaccination.
Importantly, intense WASH interventions reduce the sensitivity of the outcome to vaccination parameters.
When WASH interventions are limited, on the other hand, a nimble vaccination campaign that immunises individuals rapidly becomes an important tool in reducing total incidence.
Notably, the shape of the transition zone highlights the substitution effect between the two measures: a deficiency in sanitary infrastructure can be partially compensated by an aggressive, rapid vaccination campaign and vice versa.

%%%%%%%%%%%%%%%%%%%%
%%%%%%%%%%%%%%%%%%%%
%%%%%%%%%%%%%%%%%%%%
%%%%%%%%%%%%%%%%%%%%
\section{Discussion}
\label{sec:discussion}

Including an environmental Allee effect in the aquatic dynamics of \emph{V. cholerae} separates the disease dynamics into a local human-to-human epidemic scale and a larger aquatic ecological scale, with direct consequences for cholera control.
Traditionally, reducing the effective reproduction number below unity ($\R_v < 1$) or, in the presence of a classic backward bifurcation, below $\R_v^{SN}< 1$, guarantees disease elimination. 
However, our analysis shows the existence of an unstable boundary separating the basins of attraction of the stable states: if a sudden influx of contamination pushes the environmental pathogen concentration past a threshold linked to the Allee threshold $A_C$, the bacteria autonomously replicate toward a high endemic state above the carrying capacity $K_C$. 
Once this macroscopic endemic state is established, vaccination alone cannot eliminate the disease, underscoring the need for multi-pronged cholera control \cite{azman2020surveillance}: while sanitation and hygiene measures that reduce human shedding ($\zeta$) and waterborne exposure ($\beta_W$) shield the human population and help prevent an outbreak from crossing the Allee threshold, collapsing an already established environmental reservoir requires water treatment interventions (such as chlorination, increasing $r_C$ or directly reducing $W$) to force the system back below the $A_C$ threshold.

Because both vaccination and WASH interventions are necessary but resource-limited, our trade-off analysis examined the combined effect of WASH (reducing $\zeta$ and $\beta_W$) and reactive vaccination.
As shown in Figure~\ref{fig:heatmap-tradeoff-wash-vaccine}, weak interventions in either domain lead to substantial outbreaks, whereas a deficiency in sanitary infrastructure can be partially compensated by a rapid, aggressive vaccination campaign and vice versa.
This analysis also highlights that if WASH interventions are near-perfect, i.e., in locations with comprehensive safe water and sanitation infrastructure, waterborne transmission is effectively interrupted regardless of vaccination coverage.

In this context, it is important to carefully interpret the role of transmission from deceased individuals. While our sensitivity analysis indicates that water-mediated transmission dominates the macroscopic, sustained spread of the infection (with unburied corpses contributing substantially through post-mortem shedding into water), direct contact with infectious bodies also acts as a critical early catalyst. Empirical evidence, such as the 1994 cholera epidemic in Guinea-Bissau \cite{gunnlaugsson1998funerals}, found that eating at funerals and handling non-disinfected corpses were strongly associated with a rapid rise in cases at the start of the epidemic.
While we see in Figure~\ref{fig:plot-funeral-impulse} that a rapidly buried single body has less consequence than a single live case, it remains that anything that can help the system cross the $A_C$ threshold has the potential to generate an outbreak.

Regarding the human subsystem, although our analysis in Appendix~\ref{app:vaccine-backward-bifurcation} proves that a vaccine-induced backward bifurcation is theoretically possible in the full system, we show analytically why the environmental Allee term opposes it when $A_C > 1$ and confirm numerically across the plausible parameter space that no local backward bifurcation occurs at $\R_v = 1$.
This is reassuring from a public health perspective, as it implies that the rollout of a vaccination campaign has a predictable and strictly favourable effect on local disease control near the epidemic threshold.

Finally, we must highlight a fundamental limitation shared by this model and many others in the literature: the assumption (and technical requirement) that \emph{V. cholerae} is entirely absent at the human disease-free equilibrium.
While this assumption facilitates mathematical analysis by allowing the computation of a standard $\R_0$ or $\R_v$ and the study of bifurcations, it does not align with biological reality.
Indeed, the $W=0$ state is a necessary mathematical abstraction for isolating epidemic dynamics from the endemic ecological baseline, rather than a state that is biologically achievable.
It is well accepted that \emph{V. cholerae} is a naturally occurring aquatic organism that persists in the environment even between human outbreaks \cite{lutz2013environmental}. 
Future mathematical work should consider models in which there is no disease-free equilibrium for the bacteria, implying in turn the absence of a strict disease-free equilibrium for humans. 
Such models would operate similarly to systems with a constant immigration of infectious individuals \cite{almarashi2019effect,brauer2001models,djuikem2025transmission}; exploring this dynamic remains an important avenue for better understanding the true environmental persistence of cholera.

%%%%%%%%%%%%%%%%%%%
%%%%%%%%%%%%%%%%%%%
\subsection*{Acknowledgements}
JA is partially supported by the Natural Sciences and Engineering Research Council (NSERC) of Canada through the Discovery Grants program.

%%%%%%%%%%%%%%%%%%%
%%%%%%%%%%%%%%%%%%%
\subsection*{Data availability statement}
The code used in the computational analysis is available on GitHub at \url{https://github.com/julien-arino/overleaf-abdramane-cholera-vaccination}.

\subsection*{Ethics declarations}
\subsubsection*{Conflict of interest}
The authors declare that they have no conflict of interest.
\subsubsection*{Ethical approval}
This study did not involve any experiments with human participants or animals.
%%%%%%%%%%%%%%%%%%%
%%%%%%%%%%%%%%%%%%%
%%%%%%%%%%%%%%%%%%%
%%%%%%%%%%%%%%%%%%%
\appendix

%%%%%%%%%%%%%%%%%%%
%%%%%%%%%%%%%%%%%%%
%%%%%%%%%%%%%%%%%%%
%%%%%%%%%%%%%%%%%%%
\section{Nondimensionalisation of $W$}
\label{app:nondimensionalisation-W}

The dimensionless shedding rate $\zeta$ in \eqref{sys:general_form_dWv} is obtained from the physical \emph{per capita} shedding rate $\alpha$ as follows.
Let us denote $\alpha$ the \emph{per capita} shedding rate for humans into the environment (in cells$\times$people$^{-1}\times$day$^{-1}$), $C$ the concentration of bacteria in the water (in cells$\times$L$^{-1}$) and $Q$ the volume of water (in L).
We assume that the force of infection for transmission of the pathogen from the water occurs at the saturating rate
\begin{equation}\label{eq:infection-by-concentration-in-water}
    \lambda_W=\frac{\beta_WC}{\kappa+C},
\end{equation}
where $\kappa$ is the half-saturation constant (in cells$\times$L$^{-1}$).
To nondimensionalise and remove the dependence on water volume, one proceeds as follows.
\emph{Ceteris paribus}, the concentration $C$ of bacteria in the water is given by
\[
\frac{dC}{dt} = \frac{\alpha}{Q}I + \eta\frac{\alpha}{Q}D + r_C C \left( 1 - \frac{C}{\kappa K_C} \right) \left( \frac{C}{\kappa A_C} - 1 \right).
\]
Consider the dimensionless variable $W=C/\kappa$.
Substituting $W$ into \eqref{eq:infection-by-concentration-in-water} gives
\[
\lambda_W = \beta_W\frac{W}{1+W},
\]
which is the form in \eqref{sys:general_form_FoI}. 
For the environmental compartment, applying $C=\kappa W$ into the dynamics of $C$ and dividing both sides by $\kappa$ yields
\[
\frac{dW}{dt} = \zeta I + \eta\zeta D + r_C W \left( 1 - \frac{W}{K_C} \right) \left( \frac{W}{A_C} - 1 \right),
\]
where $\zeta=\alpha/(\kappa Q)$ is a constant (in people$^{-1}\times$day$^{-1}$). 
Thus, the dimensionless model correctly preserves the Allee dynamics, with the physical carrying capacity being $\kappa K_C$ and the physical Allee threshold being $\kappa A_C$.
It is the latter form that is used in \eqref{sys:general_form_dWv}.

%%%%%%%%%%%%%%%%%%%
%%%%%%%%%%%%%%%%%%%
%%%%%%%%%%%%%%%%%%%
%%%%%%%%%%%%%%%%%%%
\section{Basic results}
\label{app:basic-results}

The following lemma is given without proof, as it is quite classic.
It is however required in the text.
\begin{lemma}
\label{lemma:solutions-positive}
If initial conditions satisfy \eqref{sys:general_form_IC}, then solutions of \eqref{sys:general_form_v} remain non-negative for all $t\geq 0$ (and strictly positive for $t > 0$ whenever $I(0)+D(0)+W(0) > 0$ and $v > 0$).
\end{lemma}

%%%%%%%%%%%%%%%%%%%
%%%%%%%%%%%%%%%%%%%
With this in mind, we now prove Lemma~\ref{lemma:bounded-region}.
\begin{proof}[Proof of Lemma~\ref{lemma:bounded-region}]
We have
\begin{align}
     N'_H &= S'+I'+R'+V'= b -\delta I - d_H N_{H} \leq b - d_H N_H.
\end{align}
Thus,
\[
N_H(t) \leq e^{-t d_H}\left[ N_H(0) - \frac{b}{d_H} \right] + \frac{b}{d_H},
\]
which implies $\limsup_{t\to\infty} N_H(t) \leq b/d_H$, and $N_H(t) \leq b/d_H$ for all $t \ge 0$ whenever $N_H(0) \leq b/d_H$.
Using $I(t) \leq N_H(t)$, the equation for $D$ satisfies
\begin{align*}
      D' & =  \delta  I  - e  D \leq \delta N_H(t) - e D,
\end{align*}
so that $\limsup_{t\to\infty} D(t) \leq \delta b/(e d_H)$ (and $D(t) \leq \delta b/(e d_H)$ for all $t \ge 0$ when starting in $\Omega$).
Finally, for the environmental compartment, bounding $I(t)$ and $D(t)$ in the shedding terms yields, for trajectories in $\Omega$ (and asymptotically as $t \to \infty$),
\begin{align*}
      W' &= \zeta  I + \eta\zeta D + r_C W \left( 1 - \frac{W}{K_C} \right) \left( \frac{W}{A_C} - 1 \right) \\ 
         &\leq \zeta \left(1 + \eta \frac{\delta}{e}\right) \frac{b}{d_H} + r_C W \left( 1 - \frac{W}{K_C} \right) \left( \frac{W}{A_C} - 1 \right),
\end{align*}
which implies $\limsup_{t\to\infty} W(t) \leq W_{\max}$ (and $W(t) \leq W_{\max}$ for all $t \ge 0$ when starting in $\Omega$).
Consequently, $\Omega$ is attracting and positively invariant under the flow of \eqref{sys:general_form_v}.
\end{proof}

%%%%%%%%%%%%%%%%%%%
%%%%%%%%%%%%%%%%%%%
%%%%%%%%%%%%%%%%%%%
%%%%%%%%%%%%%%%%%%%
\section{The system without vaccination}
\label{app:system-without-vaccination}
Studying the system without vaccination is very helpful as many of the properties carry through to the full system.
In the absence of vaccination, \eqref{sys:general_form_v} simplifies to the following,
\begin{subequations}
		\label{sys:no-vaccination}
		\begin{align}
			\frac{d}{dt}S &= 
			b+\varepsilon R -\lambda_S  S 
			- d_H S 
			\label{sys:no-vaccination-dS} \\ 
			\frac{d}{dt}I  &= 
			\lambda_S  S -(\gamma + \delta +d_H)I  
			\label{sys:no-vaccination-dI} \\ 
			\frac{d}{dt}R  &=
			\gamma  I -(\varepsilon +d_H)R 
			\label{sys:no-vaccination-dR} \\ 
			\frac{d}{dt}D  &= 
			\delta I -e D 
			\label{sys:no-vaccination-dD} \\
			\frac{d}{dt}W  &= 
			\zeta I + \eta\zeta D + r_C W \left( 1 - \frac{W}{K_C} \right) \left( \frac{W}{A_C} - 1 \right),
			\label{sys:no-vaccination-dW} 
		\end{align}
where $N_H=S+I+R$ and the force of infection is
\begin{equation}\label{sys:no-vaccination-FoI}
\lambda_S = \beta_{W}\frac{W}{1+W}+\beta_{I}\frac{I}{N_H} +\beta_{D}\frac{D}{N_H}.
\end{equation}
\end{subequations}
As for the full model in Section~\ref{sec:math-analysis_v}, solving the equilibrium equations $S' = R' = D' = 0$ of \eqref{sys:no-vaccination} for $(S, R, D)$ in terms of $I$ and substituting them into \eqref{sys:no-vaccination-dW} and \eqref{sys:no-vaccination-dI} defines the bacterial and human nullclines $I = I_C(W)$ and $I = I_H(W)$, as well as the corresponding \emph{quasi-steady-state planar reduction} in $(W, I)$ (derived in Appendix~\ref{app:subsec:nullclines-no-vaccination}).
The following proposition summarises what is known about the equilibria of \eqref{sys:no-vaccination}.
\begin{proposition}\label{prop:no-vaccination-equilibria}
Assuming that the human-to-human reproduction number satisfies $\R_0^H := (\beta_I + \beta_D\delta/e)/(\gamma + \delta + d_H) < 1$, let $\R_0^{SN} > 1$ denote the critical saddle-node bifurcation value of $\R_0$ (varied via $\beta_W$) at which the nullclines $I_H(W)$ and $I_C(W)$ are tangent in $(0, A_C)$ (see Lemma~\ref{lemma:subthreshold-tangent}). System \eqref{sys:no-vaccination} possesses the following equilibria.
\begin{enumerate}
    \item There always exists a unique disease-free equilibrium $\bE_0^{\eqref{sys:no-vaccination}}$, which is locally asymptotically stable if $\R_0 < 1$ and unstable if $\R_0 > 1$.
    \item There always exists a unique macroscopic endemic equilibrium $\bE_m^{\eqref{sys:no-vaccination}}$ with $W^\star > K_C$, which is locally asymptotically stable in the quasi-steady-state planar reduction.
    \item If $\R_0 < 1$, there exists one subthreshold endemic equilibrium $\bE_u^{\eqref{sys:no-vaccination}}$ with $W^\star \in (0, A_C)$, which is unstable in the full system.
    \item If $1 < \R_0 < \R_0^{SN}$, there exist two subthreshold endemic equilibria with $W^\star \in (0, A_C)$: an equilibrium $\bE_\ell^{\eqref{sys:no-vaccination}}$ (locally asymptotically stable in the quasi-steady-state planar reduction) and an equilibrium $\bE_u^{\eqref{sys:no-vaccination}}$, which is unstable in the full system.
    \item If $\R_0 = \R_0^{SN}$, there exists one subthreshold endemic equilibrium $\bE_{SN}^{\eqref{sys:no-vaccination}}$ with $W^\star \in (0, A_C)$, which is a non-hyperbolic saddle-node of the full system.
    \item If $\R_0 > \R_0^{SN}$, there are no subthreshold endemic equilibria with $W^\star \in (0, A_C)$.
\end{enumerate}
\end{proposition}

\begin{proof}
    Uniqueness and local stability of the disease-free equilibrium are established in Appendix~\ref{app:subsec:no-vaccination-DFE}. 
    Existence and stability of the macroscopic equilibrium $\bE_m^{\eqref{sys:no-vaccination}}$ are proved in Lemma~\ref{lemma:one-large-W-EEP}.
    The properties of the subthreshold equilibria are respectively established by Lemmas~\ref{lemma:subthreshold-Eu}, \ref{lemma:subthreshold-two-equilibria}, \ref{lemma:subthreshold-tangent} and \ref{lemma:subthreshold-none}.
\end{proof}

%%%%%%%%%%%%%%%%%%%%%%%%%%%
%%%%%%%%%%%%%%%%%%%%%%%%%%%
\subsection{Uniqueness and stability of the disease-free equilibrium}
\label{app:subsec:no-vaccination-DFE}

Assume $I = 0$, then \eqref{sys:no-vaccination-dR} and \eqref{sys:no-vaccination-dD} imply that $R = 0$ and $D = 0$ and setting \eqref{sys:no-vaccination-dW} to zero with $I=0$ gives
\begin{equation}
    r_C W \left( 1 - \frac{W}{K_C} \right) \left( \frac{W}{A_C} - 1 \right) = 0.
\end{equation}
This yields three possible equilibria for cholera, $W \in \{0, A_C, K_C\}$. However, consider the force of infection acting on susceptible individuals, 
\[
\lambda_S = \beta_W \frac{W}{1+W} + \beta_I \frac{I}{N_H} + \beta_D \frac{D}{N_H}.
\]
If $W = A_C > 0$ or $W = K_C > 0$, then $\lambda_S > 0$. 
Since $b > 0$  ensures $S > 0$, the rate of new infections evaluates to
\begin{equation}
    \frac{dI}{dt}\bigg|_{I=0} = \lambda_S S > 0.
\end{equation}
This violates the assumption that $I = 0$. 
Thus, the \emph{unique} disease-free equilibrium must occur at $W = 0$, yielding $\bE_0^{\eqref{sys:no-vaccination}}=(b/d_H, 0, 0, 0, 0)$, where the entire population is susceptible and the environment is free of the pathogen.

Using the next-generation matrix approach as in Section~\ref{sec:math-analysis_v}, the basic reproduction number is given by
\begin{equation} \label{eq:R0-simplified}
    \R_{0}= \frac{e r_C\beta_{I}+\delta r_C \beta_{D}+\zeta (e+\eta\delta) \beta_{W} \frac{b}{d_H}}{er_C\left(\gamma +\delta +d_H \right)},
\end{equation}
which matches \eqref{eq:R0}.
Using the same type of matrix argument as in the end of Appendix~\ref{app:subsec:full-model-DFE}, it follows by \cite[Theorem 2]{VdDWatmough2002} that the disease-free equilibrium is locally asymptotically stable if $\R_0 < 1$ and unstable if $\R_0 > 1$.

%%%%%%%%%%%%%%%%%%%%%%%%%%%%%
%%%%%%%%%%%%%%%%%%%%%%%%%%%%%
\subsection{Nullclines of the system without vaccination}
\label{app:subsec:nullclines-no-vaccination}

\begin{figure}[htbp]
    \centering
    \begin{tikzpicture}
    \begin{axis}[
        width=0.8\textwidth, height=0.35\textheight,
        axis x line=center, % Draws x-axis directly at y=0
        axis y line=left,   % Draws y-axis at x=0
        xlabel={$W$},
        ylabel={$I$},
        xlabel style={at={(ticklabel* cs:1)}, anchor=north},
        ylabel style={at={(ticklabel* cs:1)}, anchor=east, rotate=-90},
        xmin=0, xmax=7.5,
        ymin=0, ymax=3.5,
        xtick=\empty,   % Hidden so we can place custom non-overlapping nodes below
        ytick=\empty,
        axis line style={thick},
        clip=false 
    ]
    % Forbidden Region (extends down into the negative quadrant)
    \fill[red!10] (axis cs:2,0.0) rectangle (axis cs:5,3.5);
    % Redraw the segment of the x-axis that gets covered by the red background fill
    \draw[thick] (axis cs:2,0) -- (axis cs:5,0);
    \node[text=red!70!black, align=center] at (axis cs:3.5, 2.8) {\small $I_C(W^\star) < 0$ \\ \small No equilibria};
    % I_C(W) Curve (Environmental Requirement)
    \addplot[domain=0:6.5, samples=200, very thick, blue] {0.8 * x * (1 - x/5) * (1 - x/2)};
    
    % I_H(W) Curve for R0 < 1
    \addplot[domain=0:7.5, samples=200, very thick, orange!90!black, dashed] {0.4 * x / (1 + 0.5*x)};
    
    % I_H(W) Curve for R0 > 1
    \addplot[domain=0:7.5, samples=200, very thick, red] {1.0 * x / (1 + 0.25*x)};
    
    % DFE Marker
    \fill[black] (axis cs:0,0) circle (2.5pt) node[anchor=south west, yshift=-20pt] {$\bE_0^{\eqref{sys:no-vaccination}}$};
    
    % Tick marks and custom axis labels for A_C and K_C 
    \draw[thick] (axis cs:2, 0.08) -- (axis cs:2, -0.08);
    \draw[thick] (axis cs:5, 0.08) -- (axis cs:5, -0.08);
    \node[anchor=north, yshift=-4pt, inner sep=1.5pt] at (axis cs:2,0) {$A_C$};
    \node[anchor=north, yshift=-4pt, inner sep=1.5pt] at (axis cs:5,0) {$K_C$};
    
    % --- Exact Intersections for R0 < 1 ---
    \fill[black] (axis cs:1.18, 0.30) circle (2.5pt);
    \draw[dotted, thick] (axis cs:1.18, 0.30) -- (axis cs:1.18, 0);
    \node[anchor=north, yshift=-4pt, fill=white, inner sep=1.5pt] at (axis cs:1.18,0) {$W_u$};
    
    \fill[black] (axis cs:5.40, 0.58) circle (2.5pt);
    \draw[dotted, thick] (axis cs:5.40, 0.58) -- (axis cs:5.40, 0);
    \node[anchor=north, yshift=-4pt, fill=white, inner sep=1.5pt] at (axis cs:5.40,0) {$W^\star_1$};
    
    % --- Exact Intersection for R0 > 1 ---
    \fill[black] (axis cs:6.18, 2.43) circle (2.5pt);
    \draw[dotted, thick] (axis cs:6.18, 2.43) -- (axis cs:6.18, 0);
    \node[anchor=north, yshift=-4pt, fill=white, inner sep=1.5pt] at (axis cs:6.18,0) {$W^\star_2$};
        
    \end{axis}
    \end{tikzpicture}
    \caption{Nullclines for the system without vaccination \eqref{sys:no-vaccination}. Blue: cholera nullcline $I_C(W)$.
    Only two of the possible nullclines $I_H(W)$ are shown: $\R_0<1$ (dashed orange) and $\R_0>\R_0^{SN}$ (red).
    See Figure~\ref{fig:allee-nullclines-novac-lt-A_C} for details on the region where $W\leq A_C$ and all possible $I_H(W)$ nullcline cases.}
    \label{fig:allee_nullclines_novac}
\end{figure}

In order to study the remaining equilibria of \eqref{sys:no-vaccination}, we conduct a geometrical analysis of the nullclines of the system.
We establish their form here to make the later reasoning easier to follow.
The situation is illustrated in Figure~\ref{fig:allee_nullclines_novac}. 
Note that this figure does not show all possible situations but is useful to understand the main features.

Setting the derivatives in \eqref{sys:no-vaccination} to zero gives that at an equilibrium $\bE^{\eqref{sys:no-vaccination}} = (S,I,R,D,W)$, there holds that
\begin{equation*}
    D = \frac{\delta}{e}I \quad \text{and} \quad R = \frac{\gamma}{\varepsilon+d_H}I.
\end{equation*}
Substituting these into the susceptible equation \eqref{sys:no-vaccination-dS} gives
\begin{equation*}
    S = \frac{b}{d_H} - \left( \frac{\delta}{d_H} + 1 + \frac{\gamma}{\varepsilon+d_H} \right)I.
\end{equation*}
The total human population at the endemic equilibrium is thus naturally bounded below $b/d_H$ by the disease-induced mortality
\begin{equation*}
    N_H = S + I + R = \frac{b}{d_H} - \frac{\delta}{d_H}I.
\end{equation*}
Because $D$, $R$, $S$ and $N_H$ are all parameterised entirely in terms of $I$, the 5-dimensional equilibrium problem reduces to finding the intersections of two curves in the $(W, I)$ plane.
The first is the \emph{cholera nullcline} $\Gamma_C(W, I) = 0$, which follows from \eqref{sys:no-vaccination-dW} by substituting $D = \delta I/e$ to get an effective shedding $\zeta_{\text{eff}} = \zeta \left(1 + \eta \delta/e\right)$, giving
\begin{equation}\label{eq:no-vaccination-I-env}
    \Gamma_C(W, I) = I - \frac{r_C}{\zeta_{\text{eff}}} W \left( 1 - \frac{W}{K_C} \right) \left( 1 - \frac{W}{A_C} \right) = 0.
\end{equation}
The second is the \emph{human nullcline} $\Gamma_H(W, I) = 0$. 
By substituting $S(I) = S_0 - c_1 I$ and $N_H(I) = S_0 - c_2 I$, where $S_0 = {b}/{d_H}$, $c_1 = {\delta}/{d_H} + 1 + {\gamma}/(\varepsilon+d_H)$ and $c_2 = {\delta}/{d_H}$, into the equilibrium incidence condition and clearing denominators, we obtain a polynomial relation explicitly in $W$ and $I$:
\begin{equation}\label{eq:no-vaccination-I-hum}
    \Gamma_H(W, I) = \beta_W W (S_0 - c_1 I)(S_0 - c_2 I) - (1+W) I (D_0 + D_1 I) = 0,
\end{equation}
where we have denoted the constants $D_0 = \left(\gamma+\delta+d_H - \beta_I - \beta_D {\delta}/{e}\right) S_0$ and $D_1 = c_1 \left(\beta_I + \beta_D {\delta}/{e}\right) - c_2 (\gamma+\delta+d_H)$ for convenience. 
Throughout the following analysis, we assume that the human-to-human reproduction number satisfies $\R_0^H := (\beta_I + \beta_D \delta/e)/(\gamma + \delta + d_H) < 1$, so that human-to-human transmission alone cannot sustain an epidemic in the absence of waterborne transmission, which ensures $D_0 > 0$.

The endemic equilibria for the system without vaccination are the intersections $\Gamma_C(W, I) = \Gamma_H(W, I) = 0$ in the positive quadrant. 
Because $\Gamma_C(W, I) = 0$ and $\Gamma_H(W, I) = 0$ are both algebraic curves of degree 3, Bézout's Theorem states that they intersect in at most 9 points in the complex projective plane.

To make the geometric analysis simpler, we can express these nullclines as functions of one variable.
From \eqref{eq:no-vaccination-I-env}, we define $I_C(W)$ as
\begin{subequations}
\label{eq:I_C-nullcline-novacc}
\begin{equation}
\label{eq:I_C-nullcline-novacc-fct}
I_C(W) = \frac{r_C}{\zeta_{\text{eff}}} W \left( 1 - \frac{W}{K_C} \right) \left( 1 - \frac{W}{A_C} \right).
\end{equation}
We have
\begin{equation}
\label{eq:I_C-nullcline-novacc-derivative}
I_C'(W) = \frac{r_C}{\zeta_{\text{eff}}} \left( 1 - 2 \left( \frac{1}{K_C} + \frac{1}{A_C} \right) W + \frac{3}{K_C A_C} W^2 \right)
\end{equation}
and
\begin{equation}
\label{eq:I_C-nullcline-novacc-second-derivative}
I_C''(W) = \frac{2 r_C}{\zeta_{\text{eff}} K_C A_C} \left( 3W - K_C - A_C \right).
\end{equation}
\end{subequations}
Note from \eqref{eq:I_C-nullcline-novacc-fct} that $I_C(W)$ has roots at $W = 0$, $W = A_C$ and $W = K_C$, with $I_C(W) > 0$ on $(0, A_C) \cup (K_C, \infty)$ and $I_C(W) < 0$ on $(A_C, K_C)$. Moreover, from \eqref{eq:I_C-nullcline-novacc-second-derivative}, $I_C''(W) < 0$ on $(0, A_C)$ (whenever $K_C > 2 A_C$, so $I_C$ is strictly concave on $(0, A_C)$) and $I_C''(W) > 0$ on $(K_C, \infty)$ (so $I_C$ is strictly convex on $(K_C, \infty)$).

For the human nullcline, we can rewrite \eqref{eq:no-vaccination-I-hum} as $W/(1+W) = \Phi(I)/\beta_W$, where
\begin{equation}
\label{eq:Phi_I_def}
\Phi(I) = \frac{I(D_0 + D_1 I)}{(S_0 - c_1 I)(S_0 - c_2 I)} = \frac{(\gamma + \delta + d_H)I}{S_0 - c_1 I} - \frac{(\beta_I + \beta_D\delta/e)I}{S_0 - c_2 I}.
\end{equation}
Under the condition $\R_0^H < 1$ ($D_0 > 0$), we have $\gamma + \delta + d_H > \beta_I + \beta_D\delta/e > 0$, $c_1 > c_2 > 0$ and $0 < S_0 - c_1 I \le S_0 - c_2 I$ for all $I \in [0, S_0/c_1)$. Differentiating \eqref{eq:Phi_I_def} term-by-term shows that $\Phi'(I) > 0$ and $\Phi''(I) > 0$ on $[0, S_0/c_1)$. Because $x \mapsto x/(\beta_W - x)$ is strictly increasing and strictly convex on $[0, \beta_W)$, the function
\begin{equation}
\label{eq:W_H-nullcline-novacc-fct}
W_H(I) = \frac{\Phi(I)}{\beta_W - \Phi(I)} = \frac{I (D_0 + D_1 I)}{\beta_W (S_0 - c_1 I)(S_0 - c_2 I) - I (D_0 + D_1 I)}
\end{equation}
satisfies $W_H(0)=0$, $W_H'(I) > 0$ and $W_H''(I) > 0$ up to the vertical asymptote $I_\infty < S_0/c_1 \le b/d_H$ where $\Phi(I_\infty) = \beta_W$. 
Consequently, its inverse $I_H(W)$ is strictly increasing ($I_H'(W) > 0$), strictly concave ($I_H''(W) < 0$) and bounded above by $I_\infty$ for all $W \ge 0$, representing the equilibrium human prevalence generated by an environmental load $W$.

%%%%%%%%%%%%%%%%%%%%%%%%%%%
%%%%%%%%%%%%%%%%%%%%%%%%%%%
\subsection{Endemic equilibria}
\label{app:subsec:no-vaccination-EEP}

First of all, notice that there can be no positive equilibria in the region shown in red in Figure~\ref{fig:allee_nullclines_novac} since the cholera nullcline $I_C(W)$ is negative there.
Therefore, further equilibria of \eqref{sys:no-vaccination} can only be found where $W<A_C$ or $W>K_C$.

%%%%%%%%%%%%%%%%%%%%%%%%%%%
\subsubsection{There is a unique equilibrium with $W^\star>K_C$}
\label{app:subsec:1EEP-gt-K_C}

We start by dealing with the simplest case, which happens in the region in Figure~\ref{fig:allee_nullclines_novac} to the right of the vertical red band.

\begin{lemma}
    \label{lemma:one-large-W-EEP}
    System \eqref{sys:no-vaccination} admits a unique equilibrium point $\bE_m^{\eqref{sys:no-vaccination}}$ where $W^\star>K_C$, which is locally asymptotically stable in the quasi-steady-state planar reduction.
\end{lemma}

\begin{proof}
    We have $I_C(K_C)=0$ and $I_C'(W)>0$ and $I_C''(W)>0$ for all $W>K_C$.
    For the human nullcline, $I_H(W)$ is strictly positive for $W>0$, so $I_H(K_C)>0$.
    Furthermore, $I_H(W)$ is strictly bounded by a maximum finite prevalence $I_\infty \le b/d_H$, whereas $I_C(W)$ defines a cubic that limits to $+\infty$ as $W \to \infty$. Since $I_H(W)$ starts above $I_C(W)$ at $W=K_C$ but is bounded as $W \to \infty$, they must intersect. The strict concavity of $I_H(W)$ and strict convexity of $I_C(W)$ on $(K_C, \infty)$ guarantee that the intersection is unique.
    As a consequence, there is a unique point of intersection when $W>K_C$ and this point always exists.

    To relate the Jacobian of the full system \eqref{sys:no-vaccination} at any endemic equilibrium $(W^\star, I^\star)$ to that of the quasi-steady-state planar reduction, write \eqref{sys:no-vaccination} in the equivalent coordinates $(\by, \bx)$ with $\by = (N_H, R, D)^T \in \IR^3$ (where $N_H = S + I + R$) and $\bx = (W, I)^T \in \IR^2$ as $\by' = F(\by, \bx)$, $\bx' = H(\by, \bx)$, where
    \[
        F(\by, \bx) = \begin{pmatrix} b - d_H N_H - \delta I \\ -(\varepsilon + d_H)R + \gamma I \\ -e D + \delta I \end{pmatrix}, \qquad \frac{\partial F}{\partial \by} = \diag(-d_H, -(\varepsilon + d_H), -e).
    \]
    Solving $F(\by, \bx) = \b0$ gives $\by = \by^\star(\bx)$ as in Appendix~\ref{app:subsec:nullclines-no-vaccination} and substituting $\by = \by^\star(\bx)$ into $H(\by, \bx)$ yields the quasi-steady-state planar reduction $\bx' = \tilde{H}(\bx) := H(\by^\star(\bx), \bx)$, namely
    \begin{equation}\label{eq:planar-reduction-novacc}
    \begin{aligned}
        W' &= f(W, I) := \zeta_{\text{eff}}\bigl(I - I_C(W)\bigr), \\
        I' &= g(W, I) := (S_0 - c_1 I)\bigl(\beta_W W/(1+W) - \Phi(I)\bigr).
    \end{aligned}
    \end{equation}
    Differentiating the identity $F(\by^\star(\bx), \bx) = \b0$ with respect to $\bx$ gives
    \[
        \frac{\partial \by^\star}{\partial \bx} = -\left(\frac{\partial F}{\partial \by}\right)^{-1}\frac{\partial F}{\partial \bx}.
    \]
    By the chain rule, the $2 \times 2$ Jacobian $J = \partial \tilde{H}/\partial \bx$ of \eqref{eq:planar-reduction-novacc} evaluated at $(W^\star, I^\star)$ is
    \[
        J = \begin{pmatrix} f_W & f_I \\ g_W & g_I \end{pmatrix} = \frac{\partial H}{\partial \bx} + \frac{\partial H}{\partial \by}\frac{\partial \by^\star}{\partial \bx} = \frac{\partial H}{\partial \bx} - \frac{\partial H}{\partial \by}\left(\frac{\partial F}{\partial \by}\right)^{-1}\frac{\partial F}{\partial \bx},
    \]
    which is the Schur complement of $\partial F/\partial \by$ in the full $5 \times 5$ Jacobian
    \[
        J_{\text{full}} = \begin{pmatrix} \partial F/\partial \by & \partial F/\partial \bx \\ \partial H/\partial \by & \partial H/\partial \bx \end{pmatrix}.
    \]
    By the Schur determinantal formula \cite[Section 0.8.5]{horn2012matrix},
    \begin{equation}\label{eq:schur-det-novacc}
        \det(J_{\text{full}}) = \det\!\left(\frac{\partial F}{\partial \by}\right)\det(J) = -d_H(\varepsilon + d_H)e\,\det(J),
    \end{equation}
    so $\operatorname{sgn}\bigl((-1)^5\det(J_{\text{full}})\bigr) = \operatorname{sgn}(\det(J))$.

    Evaluating the partial derivatives of $f$ and $g$ at $(W^\star, I^\star)$ (using $\beta_W W^\star/(1+W^\star) - \Phi(I^\star) = 0$) gives
    \begin{equation}\label{eq:planar-partials-novacc}
        f_W = -\zeta_{\text{eff}} I_C'(W^\star), \quad f_I = \zeta_{\text{eff}} > 0, \quad g_W = \frac{\beta_W(S_0 - c_1 I^\star)}{(1 + W^\star)^2} > 0, \quad g_I = -(S_0 - c_1 I^\star)\Phi'(I^\star) < 0.
    \end{equation}
    For $W^\star > K_C$, $I_C'(W^\star) > 0$ ensures $f_W < 0$, so $\mathsf{tr}(J) = f_W + g_I < 0$. 
    The slopes of the nullclines at $(W^\star, I^\star)$ are $I_C'(W^\star) = -f_W/f_I$ and $I_H'(W^\star) = -g_W/g_I$. 
    Because $I_H(W)$ crosses $I_C(W)$ from above at $\bE_m^{\eqref{sys:no-vaccination}}$, it follows that $I_C'(W^\star) > I_H'(W^\star)$. 
    Substituting the slopes into this inequality yields $-f_W/f_I > -g_W/g_I$, i.e., $f_W/f_I < g_W/g_I$.
    Multiplying both sides by the negative quantity $f_I g_I < 0$ reverses the inequality to give $f_W g_I > f_I g_W$, which rearranges to
    \begin{equation}\label{eq:planar-det-slopes}
        \det(J) = f_W g_I - f_I g_W = -f_I g_I\bigl(I_C'(W^\star) - I_H'(W^\star)\bigr) > 0
    \end{equation}
    (and $(-1)^5\det(J_{\text{full}}) > 0$ by \eqref{eq:schur-det-novacc}). 
    With $\mathsf{tr}(J) < 0$ and $\det(J) > 0$, the equilibrium point $\bE_m^{\eqref{sys:no-vaccination}}$ is locally asymptotically stable in the quasi-steady-state planar reduction.
\end{proof}

%%%%%%%%%%%%%%%%%%%%%%%%%%%
\subsubsection{Situation where $W^\star<A_C$}
\label{app:subsec:EEP-lt-A_C}

Before proceeding further, let us zoom in on the region in Figure~\ref{fig:allee_nullclines_novac} where $W<A_C$.
This is shown in Figure~\ref{fig:allee-nullclines-novac-lt-A_C}.
We start with the orange dashed curve.

\begin{figure}[htbp]
\centering
\begin{tikzpicture}
\begin{axis}[
    width=0.8\textwidth, height=0.35\textheight,
    axis x line=bottom,
    axis y line=left,
    xmin=0, xmax=11.5,
    ymin=0, ymax=4.5,
    xlabel={$W$},
    ylabel={$I$},
    xlabel style={at={(ticklabel* cs:1)}, anchor=north},
    ylabel style={at={(ticklabel* cs:1)}, anchor=east, rotate=-90},
    xtick={10},
    xticklabels={$A_C$},
    ytick=\empty,
    clip=false,
    legend cell align={left},
    legend style={at={(0.98,1)}, anchor=north east, draw=black!30, fill=white, nodes={scale=0.95}}
]

% --- Forbidden Region ---
\fill[red!10] (axis cs:10,0) rectangle (axis cs:11.5,4.5);

% Case 1: R_0 < 1 (One intersection)
\addplot[domain=0:11.5, samples=100, very thick, orange, dashed] {0.36*x / (1 + 0.1*x)};
\addlegendentry{$\R_0 < 1$}

% Case 2: R_0 > 1, Strong Saturation (Two intersections)
\addplot[domain=0:11.5, samples=100, very thick, teal] {2*x / (1 + 0.8*x)};
\addlegendentry{$1<\R_0<\R_0^{SN}$}

% Case 3: R_0 > 1, Tangent (Saddle-Node)
\addplot[domain=0:11.5, samples=100, very thick, purple, dashdotted] {2*x / (1 + 0.5828*x)};
\addlegendentry{$\R_0 = \R_0^{SN}$}

% Case 4: R_0 > 1, Low Saturation (Zero local intersections)
\addplot[domain=0:4, samples=100, very thick, red, loosely dashed] {2*x / (1 + 0.2*x)};
\addlegendentry{$\R_0 > \R_0^{SN}$}

% Cholera Nullcline (Allee bump)
\addplot[domain=0:10, samples=100, ultra thick, blue] {x - 0.1*x^2};

% --- Intersection Markers ---

% Origin
\filldraw[black] (0,0) circle (2.5pt);
\node[anchor=north west] at (0, -0.15) {$\bE_0^{\eqref{sys:no-vaccination}}$};

% Case 1 Intersection (Eu)
\filldraw[black] (8.00, 1.60) circle (2.5pt);
\node[anchor=south west, text=orange] at (8.00, 1.60) {$\bE_u^{\eqref{sys:no-vaccination}}$};

% Case 2 Intersections (E_\ell and E_u)
\filldraw[black] (1.80, 1.47) circle (2.5pt);
\node[anchor=north west, text=teal] at (1.80, 1.47) {$\bE_\ell^{\eqref{sys:no-vaccination}}$};

\filldraw[black] (6.95, 2.12) circle (2.5pt);
\node[anchor=south west, text=teal] at (6.95, 2.12) {$\bE_u^{\eqref{sys:no-vaccination}}$};

% Case 3 Tangency (Saddle-Node)
\filldraw[black] (4.14, 2.43) circle (2.5pt);
\draw[gray, thin, dotted] (4.14, 2.43) -- (4.14, 0);
\end{axis}
\end{tikzpicture}
\caption{Zoom on the region of Figure~\ref{fig:allee_nullclines_novac} where $0\leq W\leq A_C$.}
\label{fig:allee-nullclines-novac-lt-A_C}
\end{figure}

\begin{lemma}\label{lemma:subthreshold-Eu}
    If $\R_0 < 1$, \eqref{sys:no-vaccination} admits an unstable subthreshold endemic equilibrium $\bE_u^{\eqref{sys:no-vaccination}}$ with $W_u \in (0, A_C)$.
\end{lemma}

\begin{proof}
    We have $I_C'(0)=r_C/\zeta_{\text{eff}}$ and $I_H'(0)=\beta_W/\Phi'(0) = \beta_W S_0 / (\gamma + \delta + d_H - \beta_I - \beta_D\delta/e)$.
    Evaluating the derivatives of these functions at the origin, we see from \eqref{eq:R0-simplified} that $\R_0 < 1 \iff I_H'(0) < I_C'(0)$. 

    Assume $\R_0 < 1$. 
    For small $W^\star > 0$, we have $I_H(W^\star) < I_C(W^\star)$. 
    At $W^\star = A_C$, $I_C(A_C) = 0$. 
    However, because the environment is infectious, $I_H(A_C) > 0$. 
    Thus, the relationship inverts, $I_H(A_C) > I_C(A_C)$. 
    By the Intermediate Value Theorem applied to the difference $I_C(W^\star) - I_H(W^\star)$, the continuous curves must cross in the interval $(0, A_C)$. 
    
    At $\bE_u^{\eqref{sys:no-vaccination}}$, $I_H(W^\star)$ crosses $I_C(W^\star)$ from below, so the slope condition is $I_C'(W_u) < I_H'(W_u)$.
    By \eqref{eq:schur-det-novacc} and \eqref{eq:planar-det-slopes} in the proof of Lemma~\ref{lemma:one-large-W-EEP}, $I_C'(W_u) < I_H'(W_u)$ implies $\det(J) < 0$ and $(-1)^5 \det(J_{\text{full}}) < 0$ (i.e., $\det(J_{\text{full}}) > 0$), so $J_{\text{full}}$ has a strictly positive real eigenvalue and $\bE_u^{\eqref{sys:no-vaccination}}$ is unstable in the full 5-dimensional system (and is a saddle point in the planar reduction).
\end{proof}

\begin{lemma}\label{lemma:subthreshold-two-equilibria}
    If $1 < \R_0 < \R_0^{SN}$, \eqref{sys:no-vaccination} admits two subthreshold endemic equilibria $\bE_\ell^{\eqref{sys:no-vaccination}}$ and $\bE_u^{\eqref{sys:no-vaccination}}$ with $W^\star \in (0, A_C)$. 
    Furthermore, $\bE_\ell^{\eqref{sys:no-vaccination}}$ is locally asymptotically stable in the quasi-steady-state planar reduction and $\bE_u^{\eqref{sys:no-vaccination}}$ is unstable in the full system.
\end{lemma}

\begin{proof}
    The condition $\R_0 > 1$ implies $I_H'(0) > I_C'(0)$. Thus, for infinitesimally small $W^\star > 0$, we have $I_H(W^\star) > I_C(W^\star)$. 
    At $W^\star = A_C$, we again have $I_H(A_C) > I_C(A_C)$ because $I_C(A_C) = 0$ while $I_H(A_C) > 0$.
    For $1 < \R_0 < \R_0^{SN}$, the curve $I_H(W^\star)$ dips below $I_C(W^\star)$ for a bounded subset of $(0, A_C)$. 
    By the Intermediate Value Theorem, the continuous curves must cross at least twice; because the concave human nullcline $I_H(W^\star)$ saturates rapidly over the half-saturation scale $W \sim 1 \ll A_C$ while the concave arch $I_C(W^\star)$ peaks near $A_C/2$ and returns to $0$ at $A_C$, the two curves intersect at $\bE_\ell^{\eqref{sys:no-vaccination}}$ and $\bE_u^{\eqref{sys:no-vaccination}}$.
    
    At the first intersection $\bE_\ell^{\eqref{sys:no-vaccination}}$, $I_H(W^\star)$ crosses $I_C(W^\star)$ from above, so $I_C'(W^\star) > I_H'(W^\star) > 0$. 
    By \eqref{eq:planar-partials-novacc} and \eqref{eq:planar-det-slopes} in the proof of Lemma~\ref{lemma:one-large-W-EEP}, $I_C'(W^\star) > I_H'(W^\star)$ guarantees $\det(J) > 0$, while $I_C'(W^\star) > 0$ gives $f_W = -\zeta_{\text{eff}} I_C'(W^\star) < 0$ and hence $\mathsf{tr}(J) = f_W + g_I < 0$, so $\bE_\ell^{\eqref{sys:no-vaccination}}$ is locally asymptotically stable in the quasi-steady-state planar reduction.
    
    At the second intersection $\bE_u^{\eqref{sys:no-vaccination}}$, $I_H(W^\star)$ crosses $I_C(W^\star)$ from below, yielding $I_C'(W^\star) < I_H'(W^\star)$. By \eqref{eq:schur-det-novacc} and \eqref{eq:planar-det-slopes}, this implies $(-1)^5 \det(J_{\text{full}}) < 0$ (and $\det(J) < 0$), making $\bE_u^{\eqref{sys:no-vaccination}}$ unstable in the full system.
\end{proof}

\begin{lemma}\label{lemma:subthreshold-tangent}
    If $\R_0 = \R_0^{SN}$, \eqref{sys:no-vaccination} admits one subthreshold endemic equilibrium $\bE_{SN}^{\eqref{sys:no-vaccination}}$ with $W^\star \in (0, A_C)$.
\end{lemma}

\begin{proof}
    At the critical parameter value $\R_0 = \R_0^{SN}$, the curve $I_H(W^\star)$ is tangent to $I_C(W^\star)$ at a point in $(0, A_C)$.
    At this tangency point, $I_H(W^\star) = I_C(W^\star)$ and $I_H'(W^\star) = I_C'(W^\star)$. 
    By \eqref{eq:schur-det-novacc} and \eqref{eq:planar-det-slopes}, equal slopes yield $\det(J) = 0$ and $\det(J_{\text{full}}) = 0$. 
    The two equilibria $\bE_\ell^{\eqref{sys:no-vaccination}}$ and $\bE_u^{\eqref{sys:no-vaccination}}$ coalesce into a single non-hyperbolic equilibrium with a zero eigenvalue, corresponding to a saddle-node bifurcation.
\end{proof}

\begin{lemma}\label{lemma:subthreshold-none}
    If $\R_0 > \R_0^{SN}$, \eqref{sys:no-vaccination} admits no subthreshold endemic equilibria with $W^\star \in (0, A_C)$.
\end{lemma}

\begin{proof}
    For $\R_0 > \R_0^{SN}$, $I_H(W^\star)$ lies strictly above $I_C(W^\star)$ for all $W^\star \in (0, A_C)$.
    Because $I_H(W^\star) > I_C(W^\star)$ everywhere in this domain, the continuous curves never intersect and, consequently, no equilibrium point exists in this interval.
\end{proof}

\begin{remark}\label{rem:planar-vs-full}
    Note that while the stability properties of $\bE_0^{\eqref{sys:no-vaccination}}$, $\bE_u^{\eqref{sys:no-vaccination}}$ and $\bE_{SN}^{\eqref{sys:no-vaccination}}$ (as well as $\bE_\ell^{\eqref{sys:no-vaccination}}$ for $\R_0$ sufficiently close to $1$) are established for the full 5-dimensional system \eqref{sys:no-vaccination}, the local asymptotic stability of $\bE_m^{\eqref{sys:no-vaccination}}$ and of $\bE_\ell^{\eqref{sys:no-vaccination}}$ away from $\R_0 = 1$ is only established for the quasi-steady-state planar reduction \eqref{eq:planar-reduction-novacc}.
\end{remark}

%%%%%%%%%%%%%%%%%%%%%%%%%
%%%%%%%%%%%%%%%%%%%%%%%%%
%%%%%%%%%%%%%%%%%%%%%%%%%
%%%%%%%%%%%%%%%%%%%%%%%%%
\section{Equilibrium analysis and endemicity with the Allee effect}
\label{app:equilibria}

Here, we consider equilibria of the full system \eqref{sys:general_form_v} as summarised in Proposition~\ref{prop:full-vaccination-equilibria-forward}.
Much of the reasoning follows the same lines as in Appendix~\ref{app:system-without-vaccination}.

%%%%%%%%%%%%%%%%%%%%%%%%%
%%%%%%%%%%%%%%%%%%%%%%%%%
\subsection{Uniqueness and stability of the disease-free equilibrium}
\label{app:subsec:full-model-DFE}
Using the same reasoning as in Appendix~\ref{app:subsec:no-vaccination-DFE}, an infectious environment ($W \in \{A_C, K_C\}$) yields a positive force of infection $\lambda_S > 0$. 
Because the demographic recruitment $b > 0$ ensures a strictly positive pool of susceptible individuals ($S > 0, V \ge 0$), the rate of new infections evaluates to
\begin{equation}
    \frac{dI}{dt}\bigg|_{I=0} = \lambda_S S + (1-\sigma) \lambda_S V > 0.
\end{equation}
This violates the assumption that $I = 0$. 
Thus, the \emph{unique} disease-free equilibrium must occur at $W = 0$, yielding
\begin{equation}\label{eq:DFE_allee}
    \bE_0 = \left( \frac{b(\theta + d_H)}{d_H(\theta + v + d_H)}, 0, 0, \frac{bv}{d_H(\theta + v + d_H)}, 0, 0 \right).
\end{equation}
\begin{remark} Because the reproduction number $\R_v$ \eqref{eq:R0_vaccination} is derived from the spectral radius of the next-generation matrix, it depends only on the Jacobian of the system evaluated at the disease-free equilibrium $\bE_0$. 
The partial derivative of the Allee effect term with respect to $W$ evaluated at $W=0$ is $-r_C$. 
Consequently, the local cholera dynamics near the origin are governed by this linear decay rate and $\R_v$ is identical to the formulation derived for a model with simple linear bacterial clearance as is frequently used.
\end{remark}

To consider the local asymptotic stability of the disease-free equilibrium $\bE_0$, we need to check that hypotheses (A1)--(A5) of \cite[Theorem 2]{VdDWatmough2002} are satisfied.
Hypotheses (A1)--(A4) follow from the procedure used to derive $F$ and $U$ in the computation of $\R_v$ (and $-U$ is lower-triangular with negative diagonal entries $-(\gamma+\delta+d_H), -e, -r_C$, so the spectral abscissa of $-U$ is negative).
Therefore, all we need to check is that in the system without disease, the disease-free equilibrium is locally asymptotically stable.
In the absence of disease ($I=D=W=0$), \eqref{sys:general_form_v} reduces to the linear system
\[
\begin{pmatrix}
    S'\\
    R'\\
    V'
\end{pmatrix}=\begin{pmatrix}
    -v -d_H & \varepsilon & \theta\\
    0 & -\varepsilon-d_H& 0 \\
    v & 0 & -\theta - d_H
\end{pmatrix}\begin{pmatrix}
    S\\
    R\\
    V
\end{pmatrix}+\begin{pmatrix}
    b\\
    0  \\
   0
\end{pmatrix}.
\]
The matrix in this system is strictly diagonally dominant by columns with negative diagonal entries, so it is invertible and by the Gershgorin circle theorem \cite{feingold1962block} all of its eigenvalues have negative real parts (explicitly, the eigenvalues are $-d_H$, $-\varepsilon-d_H$ and $-(v+\theta+d_H)$). Thus the disease-free equilibrium is locally asymptotically stable in the disease-free subspace, verifying that assumption (A5) in \cite[Theorem 2]{VdDWatmough2002} holds.

By \cite[Theorem 2]{VdDWatmough2002}, the disease-free equilibrium is therefore locally asymptotically stable if $\R_v < 1$ and unstable if $\R_v > 1$.

%%%%%%%%%%%%%%%%%%%%%%%%%
%%%%%%%%%%%%%%%%%%%%%%%%%
\subsection{Endemic equilibria}
\label{app:subsec:full-model-EEP}
For endemic equilibria ($I^\star > 0$), we consider two nullclines very similar to those in Figure~\ref{fig:allee_nullclines_novac} but not shown here.

First, we solve \eqref{sys:general_form_dWv} to determine the human shedding required to maintain a specific equilibrium cholera load $W^\star$. 
At $\bE_\star = (S^\star, I^\star, R^\star, V^\star, D^\star, W^\star)$, substitution of $D^\star$ directly mirrors the simplified model giving the identical condition $\Gamma_C(W^\star, I^\star) = 0$, meaning
\begin{equation}
    \Gamma_C(W^\star, I^\star) = I^\star - \frac{r_C}{\zeta_{\text{eff}}} W^\star \left( 1 - \frac{W^\star}{K_C} \right) \left( 1 - \frac{W^\star}{A_C} \right) = 0.
\end{equation}
From this, we can define $I_C(W^\star)$ as the function for the environmental requirement.
Since $I^\star \ge 0$, $I_C(W^\star)$ is only biologically valid on the disconnected intervals $W^\star \in [0, A_C]$ and $W^\star \in [K_C, \infty)$. 
For $W^\star \in (A_C, K_C)$, no equilibrium can exist because the required human shedding would have to be negative to prevent a population explosion.

Second, holding $W^\star$ constant provides constant forcing on the human subsystem.
Solving the equilibrium equations for the remaining compartments yields an implicit polynomial curve $\Gamma_H(W^\star, I^\star) = 0$ for the \emph{human nullcline}. 
This curve implicitly defines $I_H(W^\star)$ as the equilibrium human prevalence generated by an environmental load $W^\star$. 
Because the water-to-human force of infection saturates ($\beta_W W^\star/(1+W^\star) \to \beta_W$ as $W^\star \to \infty$) and the total human population is finite ($N_H^\star \le b/d_H$), $I_H(W^\star)$ is a monotonically increasing function strictly bounded by a horizontal asymptote $I_{\max}$.

The endemic equilibria of \eqref{sys:general_form_v} lie at the intersections $\Gamma_C(W^\star, I^\star) = \Gamma_H(W^\star, I^\star) = 0$ in the positive quadrant.

%%%%%%%%%%%%%%%%%%%%%%%%%%%
\subsubsection{There is a unique equilibrium with $W^\star>K_C$}
\label{app:subsec:full-1EEP-gt-K_C}

\begin{lemma}
    \label{lemma:full-one-large-W-EEP}
    System \eqref{sys:general_form_v} admits a unique equilibrium point $\bE_m$ where $W^\star>K_C$, which is locally asymptotically stable in the quasi-steady-state planar reduction.
\end{lemma}

\begin{proof}
    At the carrying capacity $W^\star = K_C$, $I_C(K_C) = 0$ but $I_H(K_C) > 0$. Thus, $I_H(K_C) > I_C(K_C)$. 
    As $W^\star \to \infty$, $I_C(W^\star)$ expands as a cubic polynomial, growing unbounded toward $+\infty$. 
    Conversely, $I_H(W^\star)$ is horizontally bounded by $I_{\max}$. 
    Therefore, for sufficiently large $W^\star$, we must have $I_C(W^\star) > I_H(W^\star)$. 
    By the Intermediate Value Theorem and the strict convexity of $I_C(W^\star)$ on $(K_C, \infty)$, the continuous curves intersect at a unique point $\bE_m$ in $(K_C, \infty)$. 

    Writing \eqref{sys:general_form_v} in the coordinates $(\by, \bx)$ with $\by = (N_H, R, D, V)^T \in \IR^4$ (where $N_H = S + I + R + V$) and $\bx = (W, I)^T \in \IR^2$ as $\by' = F(\by, \bx)$, $\bx' = H(\by, \bx)$, the $4 \times 4$ Jacobian block $\partial F/\partial \by$ is lower-triangular with negative diagonal entries $-d_H$, $-(\varepsilon + d_H)$, $-e$ and $-(v + \theta + d_H + (1-\sigma)\lambda_S^\star)$, so $\det(\partial F/\partial \by) > 0$.
    By the same Schur complement calculation as in the proof of Lemma~\ref{lemma:one-large-W-EEP} \cite[Section 0.8.5]{horn2012matrix}, the $2 \times 2$ Jacobian $J$ of the quasi-steady-state planar reduction $\bx' = H(\by^\star(\bx), \bx)$ satisfies $\det(J_{\text{full}}) = \det(\partial F/\partial \by)\det(J)$ and $\operatorname{sgn}\bigl((-1)^6\det(J_{\text{full}})\bigr) = \operatorname{sgn}(\det(J)) = \operatorname{sgn}\bigl(I_C'(W^\star) - I_H'(W^\star)\bigr)$.
    Because $I_H(W^\star)$ crosses $I_C(W^\star)$ from above at $\bE_m$, we have $I_C'(W^\star) > I_H'(W^\star) > 0$, which ensures $\det(J) > 0$ and $\mathsf{tr}(J) < 0$, establishing local asymptotic stability in the quasi-steady-state planar reduction.
\end{proof}

%%%%%%%%%%%%%%%%%%%%%%%%%%%
\subsubsection{Situation where $W^\star<A_C$}
\label{app:subsec:full-EEP-lt-A_C}

Assuming the absence of a vaccine-induced backward bifurcation (i.e., $a < 0$), the human nullcline $I_H(W^\star)$ is strictly concave in this region and saturates over $W \sim 1 \ll A_C$, while the concave arch $I_C(W^\star)$ peaks near $A_C/2$ and returns to $0$ at $A_C$. 

\begin{lemma}\label{lemma:full-subthreshold-Eu}
    If $\R_v < 1$, \eqref{sys:general_form_v} admits an unstable subthreshold endemic equilibrium $\bE_u$ with $W_u \in (0, A_C)$.
\end{lemma}

\begin{proof}
    Evaluating the derivatives of these functions at the origin links the local geometry to the reproduction number: the condition $\R_v < 1$ is equivalent to $I_H'(0) < I_C'(0)$. 

    Assume $\R_v < 1$. For small $W^\star > 0$, we have $I_H(W^\star) < I_C(W^\star)$. 
    At the threshold $W^\star = A_C$, the required shedding drops to zero ($I_C(A_C) = 0$). 
    However, because the environment is infectious, human prevalence is strictly positive ($I_H(A_C) > 0$). Thus, the relationship inverts: $I_H(A_C) > I_C(A_C)$. 
    By the Intermediate Value Theorem applied to the difference $I_C(W^\star) - I_H(W^\star)$, the continuous curves must cross in the interval $(0, A_C)$. Because $I_H(W^\star)$ crosses $I_C(W^\star)$ from below ($I_C'(W_u) < I_H'(W_u)$), the Schur complement identity in Lemma~\ref{lemma:full-one-large-W-EEP} shows that $(-1)^6 \det(J_{\text{full}}) < 0$, so this intersection forms a strictly unstable equilibrium $\bE_u$ of the full 6-dimensional system.
\end{proof}

\begin{lemma}\label{lemma:full-subthreshold-two-equilibria}
    If $1 < \R_v < \R_v^{SN}$, \eqref{sys:general_form_v} admits two subthreshold endemic equilibria $\bE_\ell$ and $\bE_u$ with $W^\star \in (0, A_C)$. 
    Furthermore, $\bE_\ell$ is locally asymptotically stable (in the full system near $\R_v = 1$ and in the quasi-steady-state planar reduction) and $\bE_u$ is unstable in the full system.
\end{lemma}

\begin{proof}
    The condition $\R_v > 1$ implies $I_H'(0) > I_C'(0)$. 
    Thus, for small $W^\star > 0$, we have $I_H(W^\star) > I_C(W^\star)$. 
    At $W^\star = A_C$, we again have $I_H(A_C) > I_C(A_C)$ because $I_C(A_C) = 0$ while $I_H(A_C) > 0$.
    For $1 < \R_v < \R_v^{SN}$, the curve $I_H(W^\star)$ dips below $I_C(W^\star)$ for a bounded subset of $(0, A_C)$. 
    By the Intermediate Value Theorem, the continuous curves cross at two points $\bE_\ell$ and $\bE_u$ in $(0, A_C)$.
    
    At the first intersection $\bE_\ell$, $I_H(W^\star)$ crosses $I_C(W^\star)$ from above ($I_C'(W^\star) > I_H'(W^\star) > 0$), generating an equilibrium that is locally asymptotically stable in the full system near $\R_v = 1$ (by Theorem~\ref{thm:stability-EEP} when $a < 0$) and in the quasi-steady-state planar reduction (by the same argument as in Lemma~\ref{lemma:full-one-large-W-EEP}).
    At the second intersection $\bE_u$, $I_H(W^\star)$ crosses $I_C(W^\star)$ from below ($I_C'(W^\star) < I_H'(W^\star)$), giving $(-1)^6\det(J_{\text{full}}) < 0$ and making $\bE_u$ unstable in the full system.
\end{proof}

\begin{lemma}\label{lemma:full-subthreshold-tangent}
    If $\R_v = \R_v^{SN}$, \eqref{sys:general_form_v} admits one subthreshold endemic equilibrium $\bE_{SN}$ with $W^\star \in (0, A_C)$.
\end{lemma}

\begin{proof}
    At the critical parameter value $\R_v = \R_v^{SN}$, the curve $I_H(W^\star)$ is tangent to $I_C(W^\star)$ in $(0, A_C)$, where $I_H'(W^\star) = I_C'(W^\star)$ implies $\det(J) = 0$ and $\det(J_{\text{full}}) = 0$.
    The two equilibria $\bE_\ell$ and $\bE_u$ coalesce into a single non-hyperbolic equilibrium with a zero eigenvalue, corresponding to a saddle-node bifurcation.
\end{proof}

\begin{lemma}\label{lemma:full-subthreshold-none}
    If $\R_v > \R_v^{SN}$, \eqref{sys:general_form_v} admits no subthreshold endemic equilibria with $W^\star \in (0, A_C)$.
\end{lemma}

\begin{proof}
    For $\R_v > \R_v^{SN}$, $I_H(W^\star)$ lies strictly above $I_C(W^\star)$ for all $W^\star \in (0, A_C)$.
    Because $I_H(W^\star) > I_C(W^\star)$ everywhere in this domain, the continuous curves never intersect and, consequently, no equilibrium point exists in this interval.
\end{proof}

%%%%%%%%%%%%%%%%%%%%%%%%%%%%%%%%%%
%%%%%%%%%%%%%%%%%%%%%%%%%%%%%%%%%%
\section{The vaccination-induced backward bifurcation}
\label{app:vaccine-backward-bifurcation}

The flow diagram in Figure~\ref{fig:flow_diagram_vaccination} suggests that vaccination could induce a backward bifurcation, because of the similarity of the human part of the system to \cite{arino2003global}.
For mathematical completeness, we derive here the explicit bifurcation condition at $\R_v = 1$ and explore briefly what the consequences of such a bifurcation would be.
However, we also explain why such a bifurcation is precluded in waterborne cholera regimes with $A_C \gg 1$, meaning that the situation depicted by Proposition~\ref{prop:full-vaccination-equilibria-forward} is the one that generally holds.

\subsection{Existence of a backward bifurcation}
\label{app:existence-backward-bifurcation}
We can establish the condition for a backward bifurcation at $\R_v=1$ and the local asymptotic stability of the bifurcating endemic equilibria using \cite[Theorem 4]{VdDWatmough2002}.

\begin{theorem}
    \label{thm:stability-EEP}
    System \eqref{sys:general_form_v} exhibits a backward bifurcation at $\R_v=1$ if the bifurcation coefficient $a > 0$, with $a$ given by \eqref{eq:a_coeff_full}. 
    In this case, the subthreshold endemic equilibrium $\bE_\star$ is \emph{locally} unstable.
    If $a<0$, then the bifurcation is forward and the superthreshold endemic equilibrium $\bE_\star$ is locally asymptotically stable.
\end{theorem}

\begin{remark}
    The term \emph{locally unstable} is used to indicate that instability holds close to the bifurcation point $\R_v=1$, as far as \cite[Theorem 4]{VdDWatmough2002} is concerned.
    Proving that $\bE_\star$ is unstable whenever there are two or more subthreshold endemic equilibria requires other tools.
\end{remark}

\begin{proof}
Let the $n=6$ state variables be ordered such that the $m=3$ infected compartments appear first, $\bx = (I, D, W, S, R, V)$. 
We choose $\beta_W$ as the bifurcation parameter, letting $\beta_W^\star$ be the critical value such that $\R_v = 1$. 

Let $\bw = (w_1, \dots, w_6)^T$ be the right eigenvector of the Jacobian evaluated at $\bE_0$ corresponding to the zero eigenvalue. Solving the infected block yields:
\begin{equation}
    w_1 > 0, \quad w_2 = \frac{\delta}{e}w_1, \quad w_3 = \frac{\zeta_{\text{eff}}}{r_C}w_1.
\end{equation}
The force of infection along this vector is $\lambda^\star = \beta_I d_H w_1/b + \beta_D d_H w_2/b + \beta_W^\star w_3$. The uninfected displacement components $w_4$ and $w_6$ (for $S$ and $V$) represent the demographic shift induced by the infection and evaluate to negative displacements proportional to $w_1$.

Let $\bv$ be the corresponding left eigenvector.
Because the transitions from uninfected to infected compartments are identically zero at $\bE_0$, $v_k = 0$ for all $k > m$. 
For the $m=3$ infected components, we obtain
\begin{equation}
    v_1 > 0, \quad v_2 = \frac{1}{e}\left(v_1 \frac{d_H}{b} \beta_D S_{\text{eff}} + \eta\zeta v_3\right), \quad v_3 = v_1 \frac{\beta_W^\star S_{\text{eff}}}{r_C},
\end{equation}
where the baseline effective susceptible pool is $S_{\text{eff}} = S_0 + (1-\sigma) V_0$.

The direction of the bifurcation is determined by the sign of the bifurcation coefficient $a$. 
Because $v_k=0$ for $k>3$, the outer sum is restricted to the infected equations $f_1$, $f_2$, $f_3$,
\begin{equation}
    a = \sum_{k=1}^3 \sum_{i,j=1}^6 v_k w_i w_j \frac{\partial^2 f_k}{\partial x_i \partial x_j}(\bE_0, \beta_W^\star).
\end{equation}
Now evaluate the non-zero second derivatives for the full system equations.
\begin{itemize}
    \item For $f_2$ ($D'$), all second derivatives are identically zero.
    \item For $f_3$ ($W'$), expanding the cubic Allee effect yields $f_3 = \zeta I + \eta\zeta D - r_C W + r_C(1/K_C + 1/A_C)W^2 - r_C W^3/(K_C A_C)$. 
    Evaluated at $\bE_0$, the only non-zero second derivative is with respect to $W$,
    \begin{equation}
        \frac{\partial^2 f_3}{\partial W^2} = 2 r_C \left( \frac{1}{K_C} + \frac{1}{A_C} \right).
    \end{equation}
    \item For $f_1$ ($I'$), proportional incidence introduces non-zero second derivatives driven by the population normalisation $N_H = S + I + R + V$. 
    Evaluated at $\bE_0$, the sum of the second derivatives along the right eigenvector $\bw$ can be factored into a sum for the environmental pathway and the direct transmission pathway. 
    Letting $\lambda_H^\star = \beta_I d_H w_1/b + \beta_D d_H w_2/b$ and $\lambda_W^\star = \beta_W^\star w_3$, the full cross-derivative sum for $f_1$ evaluates to:
    \begin{align}
        \sum_{i,j=1}^6 w_i w_j \frac{\partial^2 f_1}{\partial x_i \partial x_j} &= 2 \lambda_W^\star \left(w_4 + (1-\sigma) w_6 - w_3 S_{\text{eff}}\right) \nonumber \\
        &\quad + 2 \lambda_H^\star S_{\text{eff}} \frac{d_H}{b} \left[ -w_1 - w_5 + w_4 \frac{V_0\sigma}{S_{\text{eff}}} - w_6 \frac{S_0\sigma}{S_{\text{eff}}} \right].
    \end{align}
\end{itemize}

Substituting these non-zero terms into the expression for $a$ gives
\begin{align}
    a &= v_1 \sum_{i,j=1}^6 w_i w_j \frac{\partial^2 f_1}{\partial x_i \partial x_j} + v_3 \left[ 2 w_3^2 r_C \left( \frac{1}{K_C} + \frac{1}{A_C} \right) \right]. \nonumber 
\end{align}
Substituting $v_3 = v_1 \beta_W^\star S_{\text{eff}}/r_C$ and simplifying yields the explicit condition,
\begin{align} \label{eq:a_coeff_full}
    a &= 2 v_1 \Bigg[ \lambda_W^\star (w_4 + (1-\sigma) w_6) + \beta_W^\star S_{\text{eff}} w_3^2 \left( \frac{1}{A_C} + \frac{1}{K_C} - 1 \right) \nonumber \\
    &\quad + \lambda_H^\star S_{\text{eff}} \frac{d_H}{b} \left( -w_1 - w_5 + w_4 \frac{V_0\sigma}{S_{\text{eff}}} - w_6 \frac{S_0\sigma}{S_{\text{eff}}} \right) \Bigg].
\end{align}
The transversality coefficient is
\[
    \tilde{b} = \sum_{k=1}^3 \sum_{i=1}^6 v_k w_i \frac{\partial^2 f_k}{\partial x_i \partial \beta_W}(\bE_0, \beta_W^\star) = v_1 w_3 S_{\text{eff}} > 0.
\] 

By \cite[Theorem 4]{VdDWatmough2002}, because $\tilde{b} > 0$, the condition $a > 0$ confirms a backward bifurcation where the subthreshold branch $\bE_\star$ is \emph{locally} unstable. Conversely, if $a < 0$, the bifurcation is forward.
\end{proof}

\begin{figure}[htbp]
\centering
\begin{subfigure}[b]{0.48\textwidth}
\centering
\begin{tikzpicture}
\begin{axis}[
    width=\linewidth, height=1.1\linewidth,
    axis x line=bottom,
    axis y line=left,
    xmin=0, xmax=2.5,
    ymin=-0, ymax=9,
    xlabel={$\R_v$},
    ylabel={$W^\star$},
    xlabel style={at={(ticklabel* cs:1)}, anchor=north},
    ylabel style={at={(ticklabel* cs:1.08)}, anchor=south east, rotate=0},
    xtick={0.46, 1, 1.35},
    xticklabels={$\R_v^{SN1}$, $1$, $\R_v^{SN2}$},
    ytick={4, 7},
    yticklabels={$A_C$, $K_C$},
    clip=false,
    legend style={at={(0.03,0.97)}, anchor=north west, draw=black!30}
]

% 1. DFE (Stable for Rv < 1)
\addplot[domain=0:1, samples=2, very thick, blue] {0};

% 2. DFE (Unstable for Rv > 1)
\addplot[domain=1:2.4, samples=2, very thick, blue, dashed] {0};

% 3. E_back: Subcritical emergence from origin (Unstable)
\addplot[domain=0:0.8, samples=30, very thick, blue, dashed, variable=\y] 
    ({1 - 1.5*\y + 1.2054*\y^2 - 0.2232*\y^3}, {\y});

% 4. E_micro: Micro-endemic state (Stable)
\addplot[domain=0.8:2.8, samples=60, very thick, blue, variable=\y] 
    ({1 - 1.5*\y + 1.2054*\y^2 - 0.2232*\y^3}, {\y});

% 5. E_u: Allee threshold E_u (Unstable)
\addplot[domain=2.8:4, samples=40, very thick, blue, dashed, variable=\y] 
    ({1 - 1.5*\y + 1.2054*\y^2 - 0.2232*\y^3}, {\y});

% 6. Macro-endemic E_m (Stable)
\addplot[domain=0:2.4, samples=50, very thick, blue] {7 + 0.5*x - 0.05*x^2};

% Transcritical Bifurcation dot
\filldraw[black] (1,0) circle (2.5pt);

% Saddle-Node 1 (Immunological)
\filldraw[black] (0.457, 0.8) circle (2.5pt);
\draw[dotted, thick, gray] (0.457, 0.8) -- (0.457, 0);

% Saddle-Node 2 (Ecological)
\filldraw[black] (1.35, 2.8) circle (2.5pt);
\draw[dotted, thick, gray] (1.35, 2.8) -- (1.35, 0);

% Labeling the Branches
\node[anchor=south east, text=blue] at (1, 0.3) {$\bE_{back}$};
\node[anchor=north west, text=blue] at (1, 2) {$\bE_{\ell}$};
\node[anchor=south west, text=blue] at (0.3, 3.9) {$\bE_u$};
\node[anchor=south east, text=blue] at (2.2, 8) {$\bE_m$};

\end{axis}
\end{tikzpicture}
\caption{Environmental load $W^\star$}
\end{subfigure}\hfill
\begin{subfigure}[b]{0.48\textwidth}
\centering
\begin{tikzpicture}[
    declare function={
        Rv(\y) = 1 - 1.5*\y + 1.2054*\y^2 - 0.2232*\y^3;
        Ifold(\y) = \y - 0.1*\y^2;
        Im(\x) = 30.0*\x / (1 + 3.0*\x);
    }
]
\begin{axis}[
    width=\linewidth, height=1.1\linewidth,
    axis x line=bottom,
    axis y line=left,
    xmin=0, xmax=2.5,
    ymin=-0, ymax=9,
    xlabel={$\R_v$},
    ylabel={$I^\star$},
    xlabel style={at={(ticklabel* cs:1)}, anchor=north},
    ylabel style={at={(ticklabel* cs:1.08)}, anchor=south east, rotate=0},
    xtick={0.46, 1, 1.35},
    xticklabels={$\R_v^{SN1}$, $1$, $\R_v^{SN2}$},
    ytick=\empty,
    clip=false,
]

% 1. DFE (Stable for Rv < 1)
\addplot[domain=0:1, samples=2, very thick, teal] {0};

% 2. DFE (Unstable for Rv > 1)
\addplot[domain=1:2.4, samples=2, very thick, teal, dashed] {0};

% 3. E_back: Subcritical emergence from origin (Unstable)
\addplot[domain=0:0.8, samples=30, very thick, teal, dashed, variable=\y] 
    ({Rv(\y)}, {Ifold(\y)});

% 4. E_micro: Micro-endemic state (Stable)
\addplot[domain=0.8:2.8, samples=100, very thick, teal, variable=\y] 
    ({Rv(\y)}, {Ifold(\y) - 0.15*sqrt(2.8-\y)*(\y-0.8)^2});

% 5. E_u: Allee threshold E_u (Unstable)
\addplot[domain=2.8:4, samples=100, very thick, teal, dashed, variable=\y] 
    ({Rv(\y)}, {\y - 0.1*\y^2 + 8.0*(\y - 2.8)^2 - 8.0556*(\y - 2.8)^3 + 2.0*sqrt(\y-2.8)*(4-\y)});

% 6. Macro-endemic E_m (Stable)
\addplot[domain=0:2.4, samples=50, very thick, teal] {Im(x)};

% Transcritical Bifurcation dot
\filldraw[black] (1,0) circle (2.5pt);

% Saddle-Node 1 (Immunological)
\filldraw[black] (0.457, {Ifold(0.8)}) circle (2.5pt);
\draw[dotted, thick, gray] (0.457, {Ifold(0.8)}) -- (0.457, 0);

% Saddle-Node 2 (Ecological)
\filldraw[black] (1.35, {Ifold(2.8)}) circle (2.5pt);
\draw[dotted, thick, gray] (1.35, {Ifold(2.8)}) -- (1.35, 0);

% Labeling the Branches
\node[anchor=south east, text=teal] at (1, 0.25) {$\bE_{back}$};
\node[anchor=north west, text=teal] at (1.1, 1.6) {$\bE_{\ell}$};
\node[anchor=south west, text=teal] at (0.5, 3.9) {$\bE_u$};
\node[anchor=south east, text=teal] at (2.2, 7.2) {$\bE_m$};

\end{axis}
\end{tikzpicture}
\caption{Human prevalence $I^\star$}
\end{subfigure}
\caption{Schematic bifurcation diagram depicting the equilibrium environmental pathogen load $W^\star$ (blue, left) and human prevalence $I^\star$ (teal, right) as functions of the vaccinated reproduction number $\R_v$ in the hypothetical presence of a vaccination-induced backward bifurcation at $\R_v=1$.
The disease-free equilibrium $\bE_0$ exists for both variables.
Continuous curves are locally asymptotically stable, dashed ones are unstable.
$\R_v^{SN1}$ and $\R_v^{SN2}$ denote the lower (vaccine-induced) and upper (ecological Allee) saddle-node bifurcation points, respectively.}
\label{fig:backward-bifurcation-split}
\end{figure}

Should a vaccine-induced backward bifurcation arise, the situation would be as ``caricatured'' in Figure~\ref{fig:backward-bifurcation-split}, where contrary to Figure~\ref{fig:forward-bifurcation}, we show separately the behaviour for $W^\star$ and $I^\star$ for legibility, and where two saddle-node bifurcations ($\R_v^{SN1} < 1$ and $\R_v^{SN2} > 1$, the latter corresponding to $\R_v^{SN}$ in Proposition~\ref{prop:full-vaccination-equilibria-forward}) occur along the subthreshold branch.

\begin{remark} \label{rem:bifurcation-preclusion}
Equation \eqref{eq:a_coeff_full} isolates the competing mechanisms governing the direction of the local bifurcation at $\R_v = 1$.
First, the environmental contribution contains the term $w_3^2 \beta_W^\star S_{\text{eff}} \left( {1}/{A_C} + {1}/{K_C} - 1 \right)$.
Because $W = C/\kappa$ is scaled by the half-saturation concentration $\kappa$, the condition ${1}/{A_C} + {1}/{K_C} < 1$ holds whenever $A_C > K_C/(K_C - 1) \approx 1$, i.e., whenever the physical Allee threshold $\kappa A_C$ exceeds the half-saturation concentration $\kappa$. 
Thus, an Allee threshold below the half-saturation concentration ($A_C < 1$) runs counter to the quorum-sensing regime $A_C \gg 1$, and whenever $A_C > K_C/(K_C - 1)$ the environmental Allee term is strictly negative and opposes a local backward bifurcation at $\R_v = 1$.
Second, the direct human transmission bracket in \eqref{eq:a_coeff_full} contains the positive imperfect-vaccine feedback $-w_6 S_0\sigma/S_{\text{eff}} > 0$ (since $w_6 < 0$), which is the classic mechanism capable of inducing a vaccine-driven backward bifurcation in human-to-human models \cite{arino2003global}.
However, in waterborne cholera regimes with $A_C \gg 1$, this positive feedback is overwhelmed by the negative environmental term $\beta_W^\star S_{\text{eff}} w_3^2 (1/A_C + 1/K_C - 1) < 0$ and by susceptible depletion ($w_4 + (1-\sigma)w_6 < 0$ and $-w_1 - w_5 < 0$), rendering the bifurcation coefficient $a$ strictly negative across plausible cholera parameter ranges.
\end{remark}

%%%%%%%%%%%%%%%%%%%%%%%%%%%%
%%%%%%%%%%%%%%%%%%%%%%%%%%%%
\subsection{Numerical investigation of the backward bifurcation}
\label{subsec:numerics-existence-BB}

From the expression of the bifurcation coefficient $a$ in \eqref{eq:a_coeff_full}, a vaccine-induced backward bifurcation ($a > 0$) is theoretically possible if direct human transmission and imperfect-vaccine feedback dominate, whereas Remark~\ref{rem:bifurcation-preclusion} shows that in waterborne regimes with $A_C \gg 1$ the environmental and demographic terms make $a$ strictly negative.

To test whether $a > 0$ can occur anywhere within the biologically plausible parameter space, we conducted an extensive numerical search across the 14-dimensional hyper-rectangle defined by the parameter ranges in Table~\ref{tab:parameter-values_v} (holding $b$ and $d_H$ fixed and solving for $\beta_W$).
We implemented this search by generating $1{,}000{,}000$ parameter combinations using Latin Hypercube Sampling (LHS). 
For each parameter set, we analytically solved for the critical water-to-human transmission rate $\beta_W^\star$ required to set the reproduction number to $\R_v = 1$. 

For the $386{,}834$ samples where this critical $\beta_W^\star$ fell within the biologically plausible range $[10^{-6}, 2 \times 10^{-5}]\text{ day}^{-1}$ of Table~\ref{tab:parameter-values_v}, we evaluated the bifurcation coefficient $a$ given by \eqref{eq:a_coeff_full}.
Consistent with Remark~\ref{rem:bifurcation-preclusion}, we did not find a single instance where $a > 0$, confirming that a local backward bifurcation at $\R_v = 1$ does not occur within the biologically plausible parameter ranges.

\bibliography{biblio}
\end{document}